\documentclass[nofootinbib]{article}

\usepackage[top=2cm,bottom=2cm,left=2cm,right=2cm,a4paper]{geometry}

\usepackage{amsfonts, amsmath, amsthm, graphicx, psfrag, array, bbm, float, multirow}

\usepackage{color}

\usepackage{hyperref}

\definecolor{reddish}{rgb}{0.75,0.2,0.2}

\definecolor{linkc}{rgb}{0,0,0.3}
\definecolor{rule}{rgb}{0.7,0.7,0.7}

\hypersetup{colorlinks=true, linkcolor=black, urlcolor=linkc, citecolor=linkc, bookmarksopen}

\newcommand{\be}{\begin{equation}}
\newcommand{\ee}{\end{equation}}

\newcommand{\bes}{\begin{eqnarray}}
\newcommand{\ees}{\end{eqnarray}}

\newcommand{\ba}{\begin{array}}
\newcommand{\ea}{\end{array}}

\newcommand{\cB}{\mathcal{B}}

\newcommand{\cE}{\mathcal{E}}
\newcommand{\cF}{\mathcal{F}}
\newcommand{\cG}{\mathcal{G}}

\newcommand{\cJ}{\mathcal{J}}

\newcommand{\cN}{\mathcal{N}}

\newcommand{\cR}{\mathcal{R}}

\newcommand{\cV}{\mathcal{V}}

\newcommand{\lo}{\textsc{lo}}
\newcommand{\nlo}{\textsc{nlo}}

\newcommand{\supermelon}{\textsc{supermelon}}
\newcommand{\twodipole}{\textsc{2--dipole}}

\newcommand{\onepi}{\textsc{1pi}}

\newcommand{\up}[1]{{}^{(#1)}}

\newcommand{\iid}{\textsc{iid}}
\newcommand{\sym}{\textsc{sym}}

\newcommand{\tr}{\textrm{tr}}

\newcommand{\widesim}[2][1.5]
{
  \quad\mathrel{\overset{#2}{\scalebox{#1}[1]{$\sim$}}}\quad
}

\makeatletter
\newtheorem*{rep@theorem}{\rep@title}
\newcommand{\newreptheorem}[2]{%
\newenvironment{rep#1}[1]{%
 \def\rep@title{#2 \ref{##1}}%
 \begin{rep@theorem}}%
 {\end{rep@theorem}}}
\makeatother

\newtheorem{theorem}{Theorem}[section]
\newreptheorem{theorem}{Theorem}

\newtheorem{lemma}[theorem]{Lemma}
\newreptheorem{lemma}{Lemma}

\newtheorem{proposition}[theorem]{Proposition}
\newreptheorem{proposition}{Proposition}

\newtheorem{corollary}[theorem]{Corollary}
\newreptheorem{corollary}{Corollary}

\newreptheorem{remark}{Remark}

\theoremstyle{definition}
\newtheorem{definition}[theorem]{Definition}
\newreptheorem{definition}{Definition}

\begin{document}


\title{\bf Towards a double--scaling limit for tensor models: probing sub-dominant orders}

\author{
Wojciech Kami\'nski\textsuperscript{a}, Daniele Oriti\textsuperscript{b}, James P.~Ryan\textsuperscript{b}\\[0.2cm]
\small 
a. Wydzia{\l} Fizyki, Uniwersytet Warszawski, ul Ho\.za 69, 00-681, Warsaw,
Poland\\
\small 
b. MPI f\"ur Gravitationsphysik, Albert Einstein Institute,  Am M\"uhlenberg 1, D-14476 Potsdam, Germany
 }

\date{}

\maketitle

\begin{abstract}
\noindent The definition of a double--scaling limit represents an important goal in the development of tensor models. We take the first steps towards this goal by extracting and analysing the next--to--leading order contributions, in the 1/N expansion, for the \iid\ tensor models.  We show that the radius of convergence of the \nlo\ series coincides with that of the leading order melonic sector. Meanwhile, the value of the susceptibility exponent, $\gamma_{\nlo} = 3/2$, signals a departure from the leading order behaviour.  Both pieces of information provide clues for a non--trivial double--scaling limit, for which we put forward some precise conjecture.
\end{abstract}



\section{Introduction}
\label{sec:intro}

Growing evidence is being accumulated for the (Tensorial) Group Field Theory ((T)GFT) formalism \cite{GFT1,GFT2,GFT3, vincent} as a promising overarching framework for a quantum theory of gravity; one that is able to incorporate aspects of several current discrete approaches within a powerful quantum field theory setting. TGFTs are theories of rank--$D$ tensorial fields  which generate, in their perturbative expansions, a sum over $D$--dimensional cellular (usually, simplicial) complexes. Their simplest incarnation are tensor models \cite{review, tensor},  wherein the tensors have finite index sets of size $N$.  These were proposed already in the early '90s as an attempt to reproduce, in $3d$ and $4d$, the successes of the matrix model formalism in defining both a controllable sum over topologies and a theory of random discrete geometries with a nice continuum limit (given in $2d$ by Liouville gravity). Such tensor models describe discrete geometry is purely combinatorial terms (the natural notion of distance being the graph distance on each cellular complex). Their Feynman amplitudes can thus be understood in terms of the Regge action for discrete gravity evaluated on equilateral triangulations. Moreover, the perturbative sum over Feynman diagrams coincides with the definition of quantum gravity given by the (Euclidean) Dynamical triangulations approach \cite{DT}, after appropriate identification of their respective parameter sets. When one enriches the combinatorics of tensor models with the group-theoretic data suggested by Loop Quantum Gravity \cite{lqg}, Spin Foam models \cite{SF} and simplicial geometry \cite{dupuis}, one obtains (Tensorial) Group Field Theories: proper field theories, with richer state spaces (with generic states being superpositions of spin networks) and quantum amplitudes, given by simplicial path integrals and spin foam models. It is these richer field theories, building up on the understanding of quantum geometry obtained in loop quantum gravity, that we believe offer the most promising candidates for a complete quantum theory of gravity. Actually, with the appropriate data and constructions \cite{GFT-EPRL, AristideDaniele}, TGFTs provide what can be argued to be the best fundamental definition of covariant loop quantum gravity dynamics, adapted to a simplicial context. In particular, TGFTs provide loop gravity and spin foams, as well as dynamical triangulations, with powerful, analytic field theoretic tools, suited to study of non-perturbative physics, the dynamics of many degrees of freedom, and the extraction of effective continuum geometry.

While a main motivation for TGFTs is quantum gravity, this is not their only reason of interest. TGFTs can be seen, more generally, as a new class of quantum field theories, posing interesting mathematical challenges, in particular from the axiomatic and renormalization theory perspective \cite{vincent}. At the same time, they define a new approach to statistical systems on random lattices, such as spin glasses, dimers, Ising and loop models, and in this direction they have already produced interesting results \cite{spinglasses, ising,dimers, loopgas}.

As mentioned, and whatever the perspective, their crucial asset is to provide a new setting in which unsolved problems can be tackled with the aid of powerful analytical tools from statistical and quantum field theory. In fact, many important results have been obtained in the last few years, confirming such potential. It is not the place to review  all these results \cite{review, GFT1, GFT2, GFT3}. Beyond model building of 4d gravity models, mainly from the spin foam and loop quantum gravity perspective, as well as the associated study of their quantum geometric degrees of freedom (see \cite{EPRL, AristideDaniele, EteraMaite} and references therein), work in tensor models includes: \textit{i}) a detailed understanding of the combinatorics and topology of the cellular complexes generated in perturbative expansion, which takes advantage of results in combinatorial topology \cite{FG},  concerns the absence of extended topological singularities \cite{RazvanLost}, as well as the presence of embedded Riemann surfaces \cite{jimmy}; \textit{ii}) the important identification of a large--$N$ expansion for tensor models and topological GFTs \cite{expansion1,expansion2, expansion3} (other types of large-N expansion have been proposed in \cite{expansion4, expansion5}); leading then to \textit{iii}) many further results concerning the critical behaviour of various tensor models \cite{critical, DarioRazvan} and topological GFTs; and \textit{iv}) the identification the leading order sector as branched polymers \cite{branched}. Many more results concern field theory aspects of the formalism, including universality \cite{universality, colorless}, scaling behaviour \cite{dansyl}, renormalizability \cite{vincentjoseph1,joseph2,joseph3, dansyl1, dansyl3, vignes, eterajoseph, josephvalentin, sf5}, Schwinger-Dyson equations \cite{sdequations,sdequations1,sdequations2,sdequations3} and quantum and classical symmetries \cite{joseph,AristideDanieleFlorian}, non-perturbative aspects \cite{VincentMatteo, RazvanLVE}. Finally, ways to extract effective continuum physics have been explored, for example \cite{EteraWinston,DanieleEteraFlorian,JimmyEteraDaniele,DanieleEmergentMatter, DanieleLorenzo}, culminating in the recent \cite{GFTcosmo}.

Despite all these recent successes, much remains to be done. In particular, given its crucial role in ensuring analytic control over the perturbative expansion of these models, it is important that we improve our understanding of the large--$N$ expansion of both tensor models and the more involved TGFTs. The first step is to go beyond the leading order in such an expansion, which is by now well understood, with the aim to understand the next--to--leading order (and possibly yet more sub--dominant) behaviour. In full TGFTs, subdominant processes may become dominant in certain scenarios, e.g. for particular boundary states, or with different choices of weights (that is, within different models). Even in the simplest tensor models, control over sub--dominant orders is necessary to be able to define double (and multiple) scaling limits.  While for higher--dimensional models, one should not expect two parameters to control the full series, the two parameters in the simplest \iid\ model should at least allow one to extract a broader subclass of graphs than just the leading order graphs.  In turn, they should capture better the statistical and topological properties of the sum over complexes and reveal new critical behaviour, as has been achieved in matrix models \cite{matrixdouble,tensordouble}. From the perspective of Dynamical Triangulations, the aim is to study analytically the continuum limit for finite Newton's constant.

In this paper, we study the next--to--leading order in the large-N expansion, focusing on the
simplest tensor models: the independent identically distributed (\iid) tensor model, in any
dimension. We consider this the necessary first step before tackling more involved tensor models or
TGFTs proper. To begin, we provide a brief review of such \iid\ model, the combinatorial structures
arising in their perturbative expansion, some key tools for their analysis, the large--N expansion
and the leading (melonic) order. We then move on to present the new results of our work. Our first
main result is that we identify the graphs contributing to the next--to--leading order, starting
from their core graphs.  
We show that they correspond to a precise family of graphs decorated by melons, generalizing melonic
diagrams with a single 2-dipole insertion.
  We then show that it is possible to use the Schwinger--Dyson
equations of the model to obtain a closed expression for the connected 2--point function at
next--to--leading order, as a function of the same quantity at leading order. From this, one can
extract the critical behaviour of the free energy for next--to--leading order graphs. We show that
the critical value of the coupling constant is the same as at leading order, and we identify the new
critical (susceptibility) exponent. This is our second main result. In the process, we unravel a few
more interesting technical properties of the combinatorial structures generated by tensor models.
Together with most of the technical details and proofs of the main results, they can be found in a
final appendix. We close with an extended discussion of the double scaling limit, explaining the
implications of our results for this issue, and putting forward some precise conjectures concerning
its realization.

\subsection{Random matrices, 2d dynamical triangulations and double scaling limit}
Before moving into our tensor model case, we summarize some key results from matrix models. This should serve to clarify some of our motivations as well as providing a template of what one could hope to achieve in simple tensor models.

Consider a matrix model, based on a complex $N\times N$ matrix $M$, defined by the partition function:
\begin{equation}
\label{eq:mmdef}
Z_{N,g} = \int [dM\,d\overline M] \; e^{-N\,S_{g}(M,\overline M)}\;,\quad\textrm{where} \quad
S_{g}(M, \overline M) = \tr\big(M M^\dagger\big) + \frac{\lambda}{3!}\tr\big(M^3\big) + \frac{\bar \lambda}{3!} \tr\big((M^\dagger)^3\big)\;,
\end{equation}
 $M^{\dagger} = \overline M^{T}$ and $g = \lambda\bar\lambda$. A Taylor expansion of the associated free energy per (complex) degree of freedom ($E_{N,g} \equiv (1/N^2)\log Z_{N,g}$) results in a weighted sum over connected Feynman diagrams $\cG$, dual to triangulations of orientable 2--dimensional surfaces:
\begin{equation}
\label{eq:mmTaylor}
E_{N,g} = \sum_{\mathcal{G}} \frac{1}{\sym(\cG)}\;g^{|\cV_{\cG}|/2}\; N^{ - 2h_\cG}\;,
\end{equation}
where $h_\cG$ is the genus of the surface encoded by the graph $\cG$ and $\sym(\cG)$ is a symmetry factor.\footnote{The genus of an orientable triangulated surface is defined as $h_\cG = 1 - (|\cV_\cG| - |\cE_\cG| + |\cF_\cG|)/2$, where $|\cV_\cG|$, $|\cE_\cG|$ and $|\cF_\cG|$ are the vertices, edges and faces of $\cG$, respectively.}  One can organize the graphs according to their topology in a $1/N$--expansion:
\begin{equation}
\label{eq:mmExpansion}
E_{N,g} = \sum_{h\geq 0} E_{h,g}\; N^{ - 2h}\;, \qquad\textrm{where} \qquad E_{h,g} = \sum_{\cG \;:\; h_\cG = h} \frac{1}{\sym(\cG)}\;g^{|\cV_\cG|/2}\;,
\end{equation}
and it is clear that only the $h = 0$ sector of graphs, that is the spherical triangulations, survives in the large--$N$ limit. Analyzing the series defined by $E_{0,g}$, one finds that it has a finite radius of convergence $g_c$ and leading order (non-analytic) behaviour given by:
\begin{equation}
E_{0,g} \sim \alpha_0 \left(1 - \frac{g}{g_c}\right)^{2 - \gamma} \;,\qquad \textrm{where} \qquad \gamma = -\frac{1}{2}\;,
\end{equation}
where the critical exponent $\gamma$ is known as the string susceptibility. As one tunes the coupling constant to its critical value, $g\rightarrow g_c$, a non-perurbative regime is reached, controlled by those graphs with increasingly large numbers of vertices.

A double scaling limit ensues from the fact that the series in $g$ at all orders in the $1/N$--expansion have the same radius of convergence $g_c$, although different critical behaviours given by:
\begin{equation}
\label{eq:mmCB}
E_{h,g} \sim \alpha_h \left(1 - \frac{g}{g_c}\right)^{(2 - \gamma)(1 - h)}\;.
\end{equation}
Up to issues of stability for the series in $1/N$, if one tunes:
\begin{equation}
\label{eq:mmDouble}
N \rightarrow \infty\;, \quad g\rightarrow g_c\;, \quad\textrm{with}\quad N\left( 1 - \frac{g}{g_c}\right)^{(2-\gamma)/2} = \kappa\;,
\end{equation}
where $\kappa$ is a constant, then the series $F_{N,g} \equiv N^2 E_{N,g} \sim \sum_{h} \alpha_h \; \kappa^{(2 - h)}$ includes contributions from all topologies. This is the double scaling limit.

To provide these amplitudes with a gravitational interpretation, one defines bare Newton's ($G$) and cosmological ($\Lambda$) constants through:
\begin{equation}
\label{eq:mmbare}
\log N = \frac{1}{8G}\qquad \textrm{and}\qquad \log g =   - \frac{\sqrt{3}}{2}a^2\Lambda \;,
\end{equation}
where $a$ is a length, then one may recast the weights into the following form:
\begin{equation}
\label{eq:mmEdt}
g^{|\cV_{\cG}|/2}\; N^{2 - 2h_\cG} = \exp\left[{-\frac{\sqrt{3}}{4}a^2\Lambda\, \cN_2 +
{\frac{1}{16 G}\left(2\cN_0 - \, \cN_2\right)}}\right] =
e^{-S_{G,\Lambda,a}(\cN_0, \cN_2)}\;.
\end{equation}
where $\cN_2 = |\cV_\cG|$ and $\cN_0 = |\cF_G|$ are the numbers of triangles and vertices respectively in the triangulation represented by $\cG$.  The exponent on the right hand side is the Regge action for an equilateral triangulation with edge length $a$ and thus, is of the form prescribed by the Euclidean Dynamical Triangulations (EDT) approach to 2--dimensional quantum gravity.

In this gravitational re-phrasing, the large--$N$ limit corresponds to the limit in which (the bare) Newton's constant vanishes: $G \rightarrow 0$.  Tuning the coupling constant $g$ to its critical value, and thus  to a regime controlled by those graphs with increasingly large numbers of triangles, corresponds to the large--volume limit. However, by tuning the edge length to zero simultaneously, one can obtain a continuum limit characterized by surfaces with finite macroscopic area. The expectation of the area observable $A \equiv (\sqrt{3}/{4})\,a^2\, \cN_2$ is:
\begin{equation}
\langle A \rangle
\widesim{N\rightarrow \infty}
a^2\, g\,\frac{\partial}{\partial g} \log E_{0,g}
\widesim{g\rightarrow g_c}
a^2\left(1 - \frac{g}{g_c}\right)^{-1}  \;.
\end{equation}
So tuning:
\begin{equation}
\label{eq:mmLtuning}
g\rightarrow g_c\;, \quad a\rightarrow 0\;, \quad \textrm{where}\quad a^2\left(1 - \frac{g}{g_c}\right)^{-1}  = 1/\Lambda_R\;,
\end{equation}
where $\Lambda_R$ is a renormalized cosmological constant.

The double scaling limit \eqref{eq:mmDouble}, in this perspective, has the advantage of taking into account all 2d topologies at the quantum level, but also of accessing the regime of finite Newton's constant. Indeed, it gives a renormalized constant:
\begin{equation}
\label{eq:mmGR}
\frac1{G_R} \equiv 8\log\kappa = \frac{1}{G} + 8\log \left(1 - \frac{g}{g_c}\right)^{2 - \gamma}  .
\end{equation}

\section{Tensor model essentials}
\label{sec:tensoress}

We now review the basic definitions and properties of \iid~ tensor models, their $1/N$ expansion, and the mathematical tools that are used to analyze the combinatorics and topology of their Feynman graphs.

\subsection{{\sc iid} model}
\label{ssec:iid}

Consider $D+1$ complex rank--$D$ tensors: $\phi^{i}_{n_i}$, where $i\in\{0,\dots, D\}$ is the \textit{color} of the tensor. Moreover,  each subscript $n_i$ is actually an abbreviation of the form $n_i = (n_{ii-1}, \dots, n_{i0}, n_{iD}, \dots, n_{ii+1})$, where each $n_{ij}\in\{1,\dots, N\}$ for some $N$.   The $\mathbf{(D+1)}$\textbf{--colored {\sc\iid} model} is defined by the partition function:
\begin{equation}
\label{eq:iid-pf}
Z_{N,\lambda\bar\lambda} = \int [d\phi\, d\bar\phi]\; e^{-S(\phi,\bar\phi)}\;,
\end{equation}
where $[d\phi\, d\bar\phi]$ is the Gaussian--normalized measure on each of the $(D+1)N^D$ tensor components\footnote{By Gaussian--normalized, we mean:
\be
\int [d\phi \,d\bar \phi]\; \exp\left(-\sum_{i=0}^D \sum_{n} \phi_{n_i}^i\bar\phi_{n_i}^i \right) = 1,
\ee
}
and:
\be
\label{eq:iid-action}
S(\phi,\bar\phi) = \sum_{i= 0}^D\sum_{{n}_i} \phi^i_{{n}_i} \, \bar\phi^i_{{n}_i} + \frac{\lambda}{N^{D(D-1)/4}} \sum_{{n}}\prod_{i=0}^D \phi^i_{{n}_i} +  \frac{\bar\lambda}{N^{D(D-1)/4}}\sum_{{n}}\prod_{i=0}^D \bar\phi^i_{{n}_i} \,.
\ee
It is dependent on three parameters: $\{N,\,\lambda,\,\bar\lambda\}$; the size $N$ and two coupling constants.  In the interaction terms, $\sum_{{n}}$ denotes the sum over all indices $n_{ij}$, subject to the condition that $n_{ij} = n_{ji}$. Thus, each tensor shares one argument pairwise with each of the other $D$ tensors. Let us remark briefly that this colored \iid\ model, defined in terms of $D+1$ complex tensors with simplicial interaction\footnote{The pairing of indices mimics the gluing of $(D-1)$--simplices across common faces to form a $D$--simplex.}, is equivalent to a tensor model for a single tensor.  This equivalence may be directly constructed via successive integration of all but one the tensors within the partition function of the colored simplicial model above and leads to an effective action for the remaining tensor that contains an infinite number of $U(N)^D$-invariant interactions, whose respective coupling constants are precise monomials of $\lambda\bar\lambda$ (see \cite{colorless} for details).

\subsection{1/N--expansion}
\label{ssec:expansion}

 One recognizes immediately that expressions such as \eqref{eq:iid-pf} are not na\"ively integrable, so one performs a Taylor expansion of the integrand with respect to the coupling constants, $\lambda$ and $\bar\lambda$, to obtain more manageable quantities.  One evaluates the resulting Gaussian integrals via Wick contraction. The result is summarized in a Feynman expansion as:
 \be
\label{eq:pf-expansion}
Z_{N,g} = N^D\sum_{\cG}\frac{g^p}{\sym(\cG)} N^{ -\frac{2}{(D-1)!}\omega(\cG)}, \qquad\qquad
\begin{array}{rcl}
g &=& \lambda\bar\lambda,\\[0.1cm]
p &=& \frac{|\cV_\cG|}{2},\\[0.1cm]
\omega(\cG) & = & \frac{(D-1)!}{2} \left(D + \frac{D(D-1)}{2} p - |\cF_\cG|\right),
\end{array}
\ee
 where the Feynman graphs $\cG$ label the pattern of contractions, $\textsc{sym}(G)$ is a symmetry factor, $\omega(\cG)$ is the \textbf{degree of divergence}, while $|\cV_\cG|$ and $|\cF_\cG|$ are the number of vertices and faces in $\cG$, respectively.

There are two of remarks to be made at this stage:
\begin{description}
\item[--] The Feynman graphs $\cG$ are closed $\mathbf{(D+1)}$\textbf{--colored graphs}.  This coloring allows to encode topological information, in such a way that such graphs are topologically dual to abstract simplicial D-dimensional pseudomanifolds. We will give more details on the definition and properties of colored graphs in the following. Thus, the partition function is a weighted sum over such objects.

\item[--] The degree is a non--negative graph--dependent integer. Therefore, graphs may be ordered according to their degree and since it is bounded from below by zero, it makes sense to consider a \textbf{$\mathbf{1/N}$--expansion}.

\end{description}
In fact, one can re--organize the graphs as:
\begin{equation}
\label{eq:reorganize}
Z_{N,g} = N^D  \sum_{\omega}\sum_{p}Z_{\omega,p} \, N^{-\frac{2}{(D-1)!}\omega} g^p\;,
\qquad\qquad\textrm{where}\qquad\qquad
Z_{\omega,p} = \sum_{\cG\;:\; \substack{ \omega(\cG) = \omega\\[0.05cm]  |\cV_\cG| = 2p} } \frac{1}{\sym(\cG)}\;.
\end{equation}
Calculating the coefficients $Z_{\omega,p}$ allows one to extract the critical behaviour of the series:
\begin{equation}
\label{eq:critical-series}
Z_{\omega, g} = \sum_{p} Z_{\omega,p} \, g^p\;,
\end{equation}
that is, the behaviour of the partition function at a given order in the $1/N$--expansion. With this in mind, one must label and enumerate the graphs at the order of interest, which in turn requires a more detailed examination of the $(D+1)$--colored graphs, to which we now turn.

\subsection{Essentials of (D+1)--colored graphs}
\label{ssec:graph-theory}

In this section, we present an intuitive description of some basic features of $(D+1)$--colored graphs.  For more technically precise definitions, we refer the reader to \cite{review}.

\begin{description}
\item[(D+1)--colored graphs:] A $(D+1)$--colored graph is a graph comprising of $(D+1)$-valent vertices, such that any given vertex is colored either black or white, and each of its $D+1$ incident edges is distinctly colored from the set $\{0, \dots, D\}$.  Moreover, the vertices are connected so that black vertices have only white neighbours and vice versa. An example is provided in Figure \ref{fig:colored_graph}.

\begin{figure}[htb]
\centering
\includegraphics[scale = 0.7]{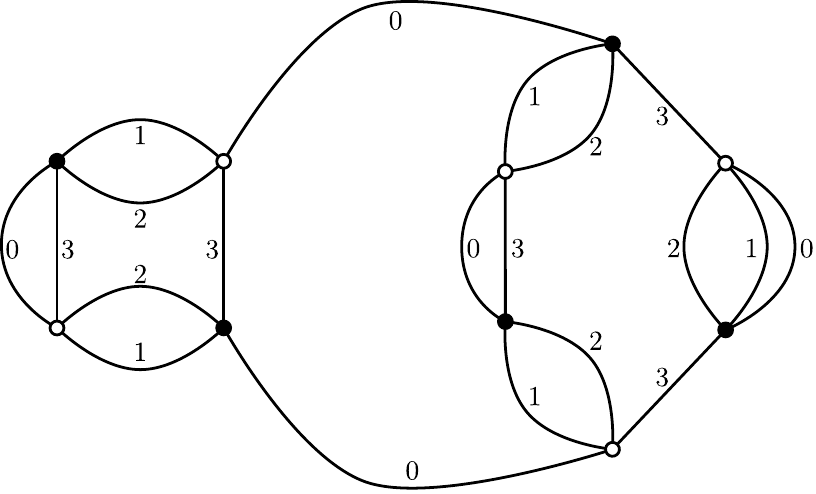}
\caption{\label{fig:colored_graph}A $(D+1)$--colored graph, with $D = 3$.}
\end{figure}

Importantly for applications to quantum theories of gravity, these colored graphs are topologically dual to $D$--dimensional abstract simplicial pseudo--manifolds.

\item[k--bubbles:] One identifies the \textbf{k--bubbles of species }$\mathbf{\{i_0,\dots,i_{k-1}\}}$ as the maximally connected subgraphs containing the $k$ distinct colors: $\{i_0,\dots, i_{k-1}\}\subset \{0,\dots, D\}$. In an obvious fashion, $k$--bubbles are nested within $(k+1)$--bubbles and so forth.  More subtly, the $k$--bubbles are dual to the $(D-k)$--dimensional simplices in the associated simplicial complex, while the nesting relations encode how these simplices are glued together.  Some $k$--bubbles of Figure \ref{fig:colored_graph} are identified in Figure \ref{fig:colored_bubbles}.

\begin{figure}[htb]
\centering
\includegraphics[scale = 0.7]{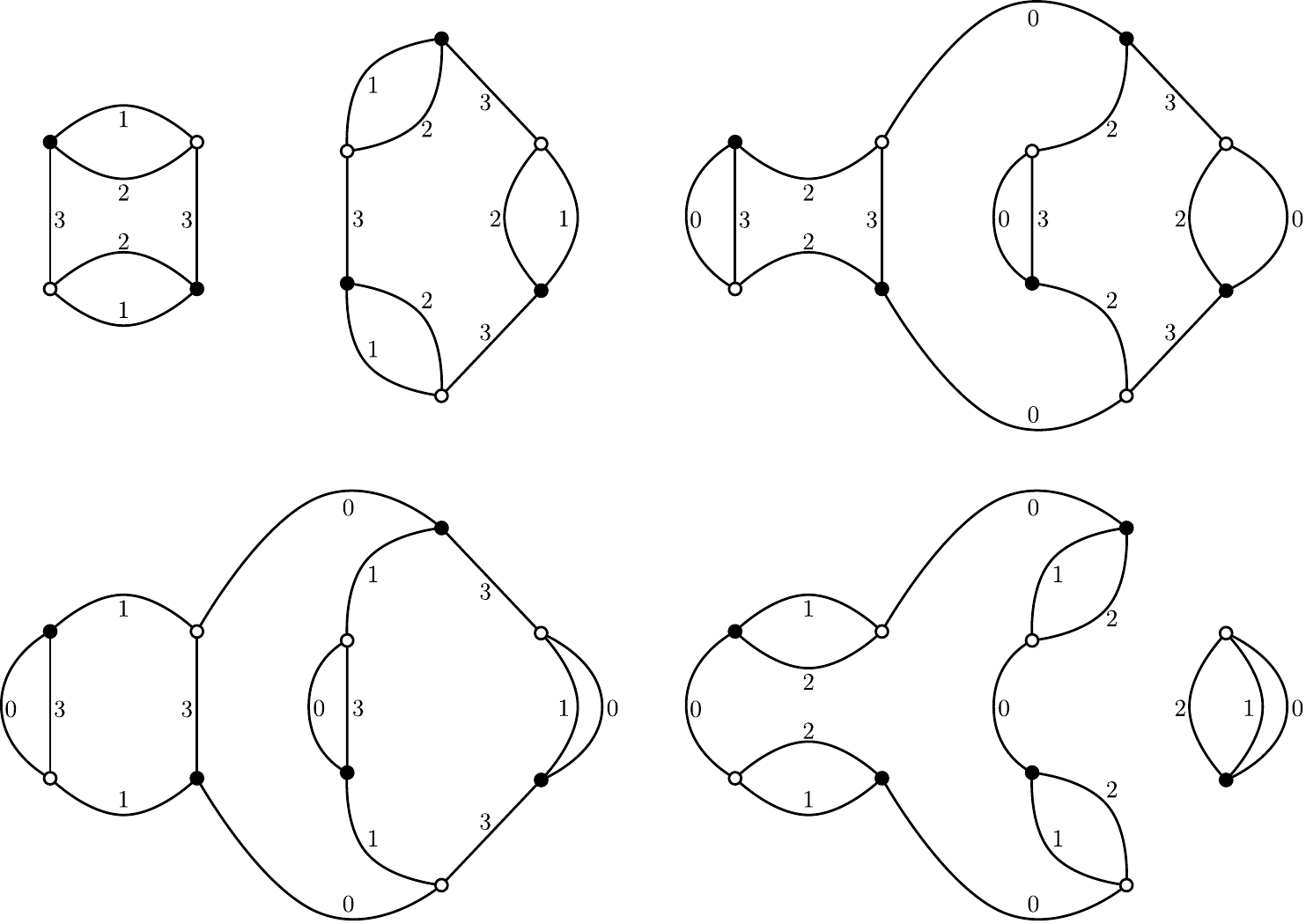}
\caption{\label{fig:colored_bubbles}The 3--bubbles of species $\{\widehat 0\}$, $\{\widehat 1\}$, $\{\widehat 2\}$ and $\{\widehat 3\}$ (clockwise).}
\end{figure}

In particular, note that the faces are the 2--bubbles. We shall often use the notation: $\{\widehat i_0, \dots,\widehat i_{k-1}\} = \{0,\dots, D\}\backslash \{i_0,\dots,i_{k-1}\}$.

It is possible to express the degree of a graph in terms of the degree of its $(k+1)$--bubbles:
\begin{equation}
\label{eq:degreeKbubble}
\omega(\cG) =\frac{D!}{2} +  \frac{(D+1)!}{2}\left(\frac{1}{k+1} - \frac{1}{D+1}\right)p - \frac{k!(D-k)!}{2} B^{[k+1]} + \sum_{(i_0\dots i_k;\rho)} (D-k)!\; \omega(\cB_{(i_0\dots i_k;\tau)}) \;,
\end{equation}
where $(i_0\dots i_k;\rho)$ labels distinct $(k+1)$-bubbles in $\cG$ and $B^{[k+1]}$ denotes the total number of $(k+1)$-bubbles in $\cG$.
which in the case of its $D$--bubbles, reduces to the following relation:
\begin{equation}
\label{eq:degreeDbubble}
\omega(\cG) = \frac{(D-1)!}{2}\left(p + D - B^{[D]}\right) + \sum_{(\hat i;\rho)} \omega(\cB_{(\hat i;\rho)})\;.
\end{equation}

\item[k--dipoles:]  One wishes to catalogue graphs. For colored graphs, the key tool to do so is a class of combinatorial moves that have a well--controlled effect on bubble structure.  These transformations are known as \textbf{k--dipole move}s.  A $k$--dipole move of species $\{i_0,\dots,i_{k-1}\}$ is illustrated in Figure \ref{fig:k-dipole}. As one can see, there are actually two types of dipole moves, \textbf{dipole creation} and \textbf{dipole annihilation}, one being the inverse of the other. A $k$--dipole annihilation consists, roughly speaking, in the removal of $k$ lines connecting a white and a black vertex, together with the vertices themselves, while joining the remaining $D+1-k$ lines. There is one condition that must be satisfied by the $k$ edges of colors $i_1,\dots,i_k$ on the right of Figure \ref{fig:k-dipole} -- they should separate two distinct $(D+1-k)$--bubbles of species $\{\widehat i_0, \dots,\widehat i_{k-1}\}$.  Thus, $k$--dipole creation (annihilation) increases (resp. decreases) the number of $(D+1-k)$--bubbles of species $\{\widehat i_0,\dots, \widehat i_{k-1} \}$ by 1.

\begin{figure}[h]
\centering
\includegraphics[scale=0.9]{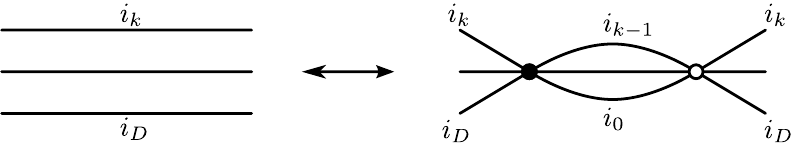}
\caption{\label{fig:k-dipole} The k--dipole moves of species $\{i_0,\dots, i_{k-1}\}$.}
\end{figure}

Additionally, if the $(D+1-k)$--bubble added (removed) is a $(D-k)$--sphere, then the $k$--dipole implements a homeomorphism on the associated topological space.

Then, consider the following example where $\cG$ possesses a $k$--dipole of some species and the graph resulting from the annihilation of this dipole is denoted by $\cG\backslash d_k$.  Their respective degrees are related by:
\begin{equation}
\label{eq:degreedipole}
\omega(\cG) = \omega(\cG\backslash d_k) + \frac{(D-1)!}{2}(k - 1)(D - k)\;.
\end{equation}
Importantly, for both $k =1$ and $k=D$ the degree is unchanged.

\item[1--dipoles and core graph equivalence classes:]  Given that 1--dipole moves preserve the degree, it is perhaps unsurprising that they play a special role in cataloguing Feynman graphs of the \iid\ model.  One can partition the graphs with a given degree into equivalence classes, where the 1--dipole moves constitute the equivalence relation. In other words, two graphs are in the same equivalence class if they are related by a sequence of 1--dipole moves (both of creation and annihilation type).   Furthermore, it emerges that these equivalence classes come equipped with convenient representatives, known as \textbf{core graphs} -- those members of the class from which no more 1--dipoles can be annihilated. In general, there are several core graphs within a particular equivalence class.  However, this does not pose a problem.  One simply picks one such graph for each equivalence class.  The rest of the graphs in the class are generated by performing arbitrary sequences of 1--dipole moves on this core graph.

One would also like to label each graph in the equivalence class uniquely, so that the coefficients in \eqref{eq:reorganize} can be computed.  This turns out to be tricky and, unfortunately, the sequences of 1--dipole moves mentioned a moment ago are not the best tool to achieve this goal. The reason is that often there are several distinct sequences that transform a representative core graph to the same graph in the equivalence class. The next section details the solution to this problem in the leading order and next--to--leading order sectors.

\item[Jackets:] Jackets are the name given to a certain class of 2--dimensional surfaces embedded within the $D$--dimensional topological manifold.  They are encoded via a $(D+1)$--cycle $\sigma$ of the set $\{0,\dots,D\}$.  The surface is constructed from the cycle as follows.  Consider a graph $\mathcal{G}$ and a planar projection of the neighbourhood of each black vertex such that the incident colored edges are ordered clockwise around the vertex according to the cycles $\sigma$.  For white vertices, the cycle determines the anti--clockwise ordering.  Since a $(D+1)$--colored graph is not generally planar, away from the vertices the edges cross and the graph is embedded in a surface of non--zero genus.  This surface is provided by  $\mathcal{G}$ itself,  comprising of the totality of its vertices and edges, along with all the faces whose two colors are adjacent in the cycle $\sigma$.

There are $D!/2$ distinct jackets $\mathcal{J}$ in a graph.\footnote{Although there are $D!$ distinct cycles for a set with $D+1$ elements, reversing the cycle does not does not generate a different surface.}
 As a result, the degree may be re--expressed as a sum over the genera of the jackets:\footnote{For an orientable surface: $h_{\cJ} = 1 - (|\cV_\cJ| - |\cE_\cJ| + |\cF_\cJ|)/2$, where  $|\cV_\cJ|$, $|\cE_\cJ|$, $|\cF_\cJ|$ are respectively the vertices, edges and faces of the jacket $\cJ$.}
\begin{equation}
\label{eq:degreejacket}
\omega(\cG) = \sum_{\cJ} h_{\cJ}\;.
\end{equation}

\item[Gravitational interpretation:]

As in the 2d case outlined in the introduction, the amplitudes have a gravitational interpretation in terms of the Regge action evaluated on equilateral triangulations\footnote{In proper TGFTs, on the other hand, thanks to their richer set of data, the correspondence can be improved to give generic simplicial path integrals for discrete (1st order) gravity actions, with generic assignment of geometric variables (areas of triangles, holonomies of discrete gravity connections, etc), in turn dual to spin foam models \cite{GFT1,GFT2, GFT-EPRL, EPRL, AristideDaniele}.}
, and the sum over graphs can thus be put in correspondence with the definition of quantum gravity suggested by the Euclidean Dynamical Triangulations approach. Defining:
\begin{equation}
\label{eq:ttbare}
\log N = \frac{\cV_{D-2}}{8G}\;,\quad \log g =  \frac{D}{16\pi G}\cV_{D-2}\left((D-1)\pi - (D+1)\arccos \frac{1}{D}\right) - 2\cV_{D}\Lambda \;,
\end{equation}
one may recast the weights as:
\begin{equation}
\label{eq:ttEdt}
\begin{split}
g^{|\cV_{\cG}|/2}\; N^{D - \frac{2}{(D-1)!}\omega(\cG)} &=
\exp\left[- \cV_{D}\,\Lambda\, \cN_D + \frac{1}{16\pi G}\left(2\pi \cV_{D-2}\, \cN_{D-2} - \frac{D(D+1)}{2}\cV_{D-2}\,\cN_{D}\arccos \frac{1}{D}\right)\right]\\
 &= e^{-S_{G,\Lambda,a}(\cN_{D-2}, \cN_D)}\;.
 \end{split}
\end{equation}
where $\cN_D = |\cV_\cG|$ and $\cN_{D-1} = |\cF_G|$ are the numbers of $D$-- and $(D-2)$--simplices respectively in the triangulation represented by $\cG$, while $\cV_k = (a^k/k!)\sqrt{(k+1)/2^k}$ is the volume of an equilateral $k$--simplex with edge--length $a$. Once again, this is the action prescribed by the EDT approach. The large--$N$ limit corresponds to the vanishing of (the bare) Newton's constant: $G \rightarrow 0$, while tuning the coupling constant to its critical value lead to a regime whose behaviour is controlled by $D$--complexes with increasingly large numbers of $D$--simplices. As in the 2d case, a double scaling limit would then allow not only to include a more general class of triangulations in the sum (although in this case it may not allow to go beyond spherical topology), but also to probe the regime corresponding to finite Newton's constant.
\end{description}

\subsection{Observables}
\label{ssec:observables}

Rather than deal with the partition function directly, two other observables are studied in this paper. First, the free energy is defined as:
\begin{equation}
\label{eq:free-energy}
E_{N,g} = \frac{1}{N^D} \log Z_{N,g}\;,
\end{equation}
and its contributions come from {\it connected} closed $(D+1)$--colored graphs.
Meanwhile, the connected 2--point function is defined as:
\begin{equation}
\label{eq:2-point-function}
\langle \phi^i_{m}\,\bar \phi^i_{\bar m}\rangle_c = \frac{1}{Z_{N,g}}\int [d\phi\,d\bar\phi]\; \phi^i_m \,\bar \phi^i_{\bar m}\; e^{-S(\phi,\bar\phi)}\;,
\end{equation}
and its contributions come from connected $(D+1)$--colored graphs with two external edges of color $i$. Since all connected closed $(D+1)$--colored graphs are also 1--particle irreducible (\onepi), cutting a single edge of color $i$ within \textit{any} connected closed graph gives a connected 2--point graph.
Moreover, for a closed graph $\cG$ scaling like $N^{D - \frac{2}{(D-1)!}\omega(\cG)}$, its associated 2--point graph, $\widetilde \cG$, scales like $N^{-\frac{2}{(D-1)!}\omega(\cG)}$. In other words, all graphs get rescaled by the same factor $N^{-D}$.  Thus, a 2--point graph contributes at a certain order to the 2--point function if and only if its associated closed graph contributes at that order to the free energy.

Due to index conservation, the connected 2--point function may be factorised as:
\begin{equation}
\label{eq:factorise-2-point}
\langle \phi^i_{m}\,\bar \phi^i_{\bar m}\rangle_c =G_{N,g} \, \delta_{m\bar m}\;,
\end{equation}
where the factor $G_{N,g}$ is independent of the color of the external edges.
Indeed, there exists a Schwinger--Dyson equation relating these observables:\footnote{The appropriate Schwinger--Dyson equation is:
\begin{equation}
\label{eq:sd-equation-proto}
0 = \int [d\phi\,d\bar\phi] \frac{\partial}{\partial \phi^i_{\bar m}} \left( \phi^i_m\, e^{-S(\phi,\bar\phi)}\right)\; \nonumber.
\end{equation}
}
\begin{equation}
\label{eq:2pointSD}
G_{N,g} = 1 + g\frac{\partial}{\partial g} E_{N,g}\;.
\end{equation}
Thus, the behaviour of the free energy is directly and easily related to that of the connected 2--point function.  This is a very useful property for the analysis to be detailed below, since connected 2--point graphs are more easily catalogued than closed graphs.

\section{Graphs}
\label{sec:graphs}

We now present the analysis of the combinatorial structure of the graphs appearing in the perturbative expansion \ref{eq:pf-expansion}, in the $1/N$ expansion, at both leading and next-to-leading order. We focus on the main steps of the analysis and on the results, leaving the detailed proofs to the appendix.

\subsection{Leading order}
\label{ssec:lo}

We shall just state the results obtained in \cite{critical}.  One finds that the leading order in the $1/N$--expansion is specified by: $\omega(\cG)=0$. Thus, the pertinent coefficients are the:
\begin{equation}
\label{eq:lo-coeff}
G_{\lo, p} := G_{0,p} = \sum_{\widetilde\cG\;:\; \substack{\omega(\cG) = 0\\[0.05cm] |\cV_{\widetilde \cG}| = 2p}} \frac{1}{\sym(\widetilde\cG)}.
\end{equation}
where $\cG$ is the closed graph obtained from the 2--point graph $\widetilde \cG$ by joining its two external lines of color $0$.

For the leading order closed graphs, there is a single equivalence class with a unique core graph called the \textbf{supermelon}. It is illustrated in Figure \ref{fig:supermelon}.

\begin{figure}[htb]
\centering
\includegraphics[scale = 0.9]{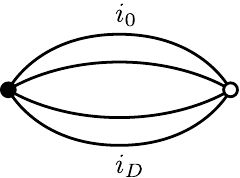}
\caption{\label{fig:supermelon} The supermelon core graph, denoted $\cG_{\supermelon}$.}
\end{figure}

Then, one turns to the 2--point graphs.  The graphs occurring at leading order are known as \textbf{rooted melonic graph}s.  The fundamental building blocks of any rooted melonic graph are the \textbf{elementary melon}s, illustrated in Figure \ref{fig:elmelon}.

\begin{figure}[htb]
\centering
\includegraphics[scale=0.9]{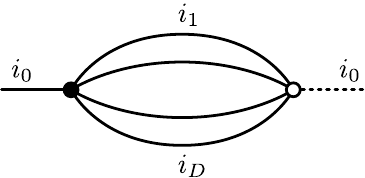}
\caption{\label{fig:elmelon} An elementary melon of species $i_0$.  }
\end{figure}

 Such a melon consists of two vertices sharing $D$ edges. Both vertices have one external edge.  Obviously, both external edges possess the same color, say $i_0$. Thus, one refers to such an object as an elementary melon of color $i_0$. Moreover, an elementary melon has two distinguished features: \textit{i}) an external edge of color $i_0$ incident to the white vertex, which is known as the \textbf{inactive edge}; \textit{ii})  $D+1$ edges incident at the black vertex, which are known as \textbf{active edges}.

The set of all rooted melonic graphs, denoted by $M$, is the union of the subsets $M_p$ containing rooted melonic graphs with $2p$ vertices:
\begin{equation}
\label{eq:melonUnion}
M = \bigcup_{p \geq 1} M_p.
\end{equation}
One may define the elements of $M_p$ as follows:
\begin{description}
\item[p = 1:]  There are only $D+1$ rooted melonic graphs in $M_1$, the elementary melons illustrated in Figure \ref{fig:elmelon} for different choices of  $i_0\in\{0,\dots,D\}$.

\begin{figure}[H]
\centering
\includegraphics[scale=0.8]{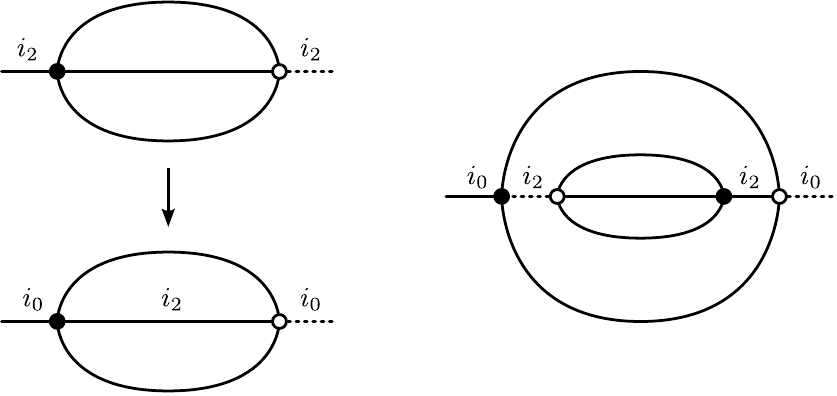}
\caption{\label{fig:elMelonInsertion} An elementary melon of color $2$ inserted along the active edge of color $2$ (for $D=3$). The active edges are drawn using full lines.}
\end{figure}

\item[$p=2$:]  One obtains the graphs in $M_2$ from the graphs in $M_1$ by
replacing an active edge of a given color by an elementary melon of the same color, as shown in Figure \ref{fig:elMelonInsertion}. (In fact, this is $D$--dipole creation.)

\item[$p=k$:] One obtains the graphs in $M_p$ from those in $M_{p-1}$ by replacing some active edge by an elementary melon.

\end{description}

For a graph occurring in $M_p$, the initial combinatorial factor coming from the Taylor expansion is $1/p!$, while the graph is obtained from exactly $p!$ Wick contractions. Thus, the final combinatorial factor is $\sym(\widetilde\cG) = 1$.   As a result, the problem of calculating the coefficients $G_{\lo,p}$ has been reduced to the enumeration of distinct patterns of melonic insertions.

\noindent\textbf{Remark:} A generic leading order 2--point graph is obtained from a leading order 2--point graph at $p =1$ by performing some sequence of 1--dipole moves, just as a generic leading order closed graph is obtained from $\cG_{\supermelon}$ by some sequence of 1--dipole moves.   However, in our definition of rooted melonic graphs, we have only talked about inserting melons within melons.  This stems from the result, proven in \cite{critical}, that: \textit{The set of rooted melonic graphs $M$ is closed under 1--dipole creation and annihilation.}


\subsection{Next--to--leading order: statement of the results}
\label{ssec:nlo}

In this section, we present the results of our analysis of the next--to--leading order sector.  One can anticipate its components: \textit{i}) the identification of \nlo\ core graphs; \textit{ii}) an iterative procedure to generate all graphs at that order, starting from the core graphs. For the technical aspects, we refer the reader to Appendix \ref{app:nloTech} where the precise statements are laid out and proven.

Our first main result concerns the core graphs:
\begin{proposition}
\label{prop:core}
Consider the \iid\ tensor model with $D\geq 3$. The graphs $\cG_{\twodipole}$ (seen in Figure \ref{fig:core-graphs-nlo}) are the \nlo\ core graphs.

\begin{figure}[htb]
\centering
\includegraphics[scale = 0.9]{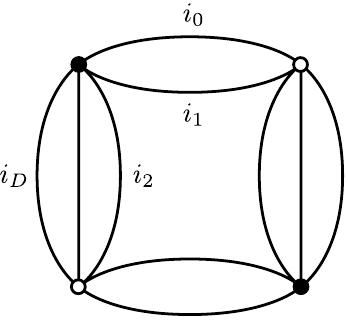}
\caption{\label{fig:core-graphs-nlo} The core graphs at \nlo, denoted collectively by $\cG_{\twodipole}$.}
\end{figure}

\end{proposition}

In the \nlo\ sector, there are ${D+1 \choose 2}$ core graphs, all of the form given in Figure \ref{fig:core-graphs-nlo}. Note that they may be obtained  from the supermelon graph by creating a single 2--dipole, illustrated in Figure \ref{fig:twoDipole}, of which there are ${D+1 \choose 2}$ distinct species.

\begin{figure}[htb]
\centering
\includegraphics[scale=0.8]{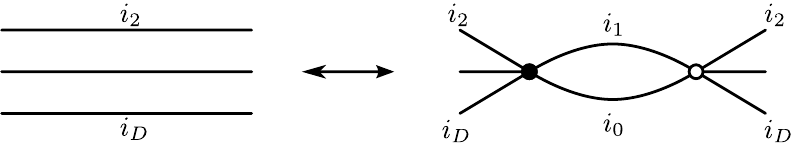}
\caption{\label{fig:twoDipole} A 2--dipole of species $\{i_0i_1\}$.}
\end{figure}

It emerges that the creation of the first 1--dipole in the \nlo\ core graphs is equivalent to the insertion of an elementary melon.  However, the creation of a second 1--dipole has two possible effects: \textit{i}) it may again be equivalent to the insertion of an elementary melon or \textit{ii}) it may  produce a graph of the form illustrated in Figure \ref{fig:extension}.

\begin{figure}[htb]
\centering
\includegraphics[scale=0.8]{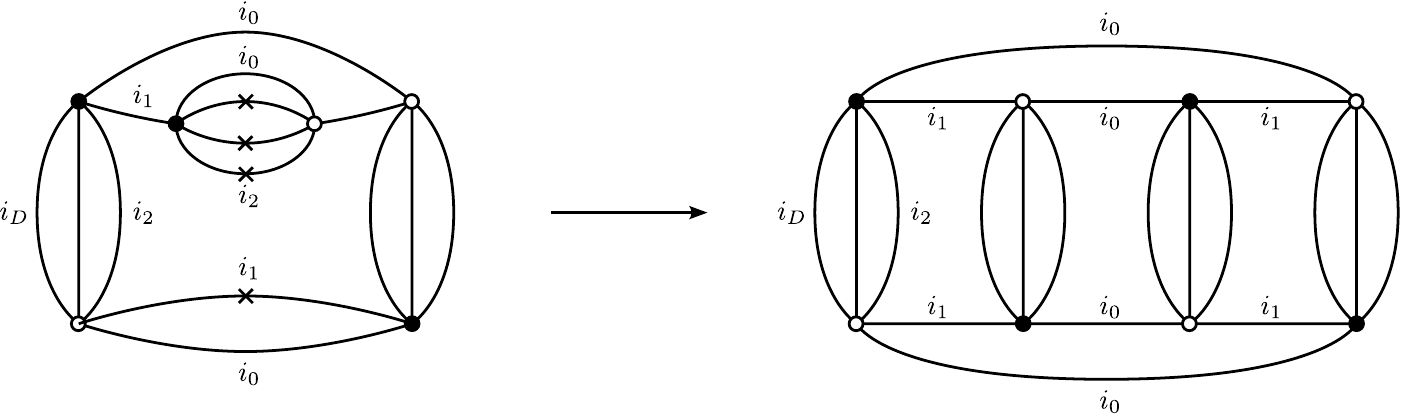}
\caption{\label{fig:extension} The creation of a second 1--dipole, of species $\{i_0\}$, which is not equivalent to the insertion of an elementary melon. The edges contributing to the 1--dipole are tagged.}
\end{figure}

One may iterate this procedure to arrive at graphs of the form drawn in Figure \ref{fig:extensionG}. Note that the graphs produced by $\ell-1$ iterations of this procedure have two faces of species $\{i_0i_1\}$, each with $2\ell$ edges. We shall denote such graphs by $\cG_{\twodipole,\ell}\;$, in which case, $\cG_{\twodipole,1} \equiv \cG_{\twodipole}$.

\begin{figure}[htb]
\centering
\includegraphics[scale=0.8]{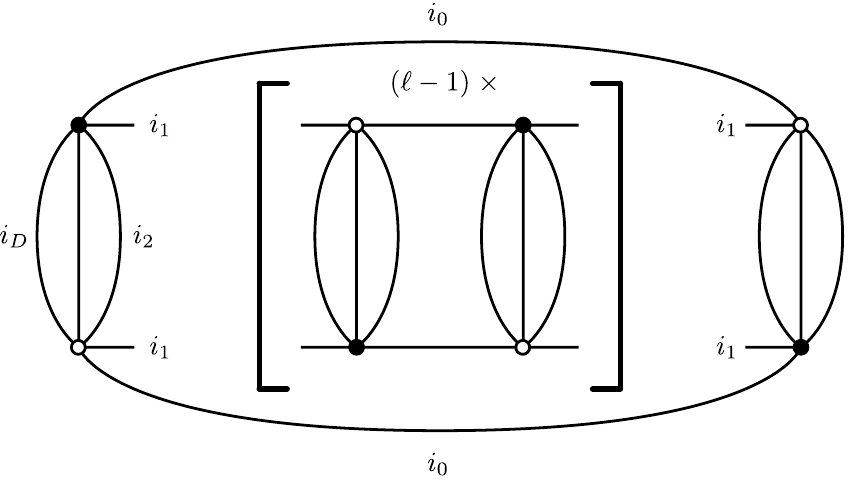}
\caption{\label{fig:extensionG} The graphs $\cG_{\twodipole,\ell}$.}
\end{figure}

\begin{definition}
\label{def:closedgraphs}
One denotes by $S_{\ell,n}$, the set of closed graphs derived from $\cG_{\twodipole,\ell}$ by inserting arbitrary combinations of $n$ elementary melons.
\end{definition}

\begin{proposition}
\label{prop:nlographsClosed}
The set of graphs:
\begin{equation}
S = \bigcup_{\substack{\ell \geq 1\\[0.05cm] n\geq 0}} S_{\ell,n}
\end{equation}
is closed under 1--dipole creation and annihilation.
\end{proposition}
As at leading order, it is easier to accurately count distinct \nlo\ connected 2--point graphs rather than closed graphs. Given that the next--to--leading order sector in the $1/N$--expansion is specified by: $\omega(\cG) = \omega(\cG_{\twodipole}) =\frac{(D-1)!}{2} (D-2)$, one finds that:
\begin{equation}
\label{eq:leading-order}
G_{\nlo,p} := G_{\frac{(D-1)!}{2} (D-2),p} =  \sum_{\widetilde \cG \;:\; \substack{\omega(\cG) = \frac{(D-1)!}{2} (D-2)\\[0.1cm]  |\cV_{\widetilde\cG}| = 2p}} \frac{1}{\sym(\widetilde\cG)}\;.
\end{equation}

There are a number elementary building blocks that are used to define 2--point graphs at this order.  Supplementing the elementary melons, there are the 2--point insertions obtained by cutting a edge of $\cG_{\twodipole,\ell}$ (with $\ell \geq 1$).  Depending on the edge cut, they take the forms illustrated in Figure \ref{fig:el2dipole} and we shall refer to them all as \textbf{elementary 2--dipole}s.  Note that, once again, the solid edges are active, while there is still just one inactive edge, marked by a dashed line.

\begin{figure}[htb]
\centering
\includegraphics[scale=0.8]{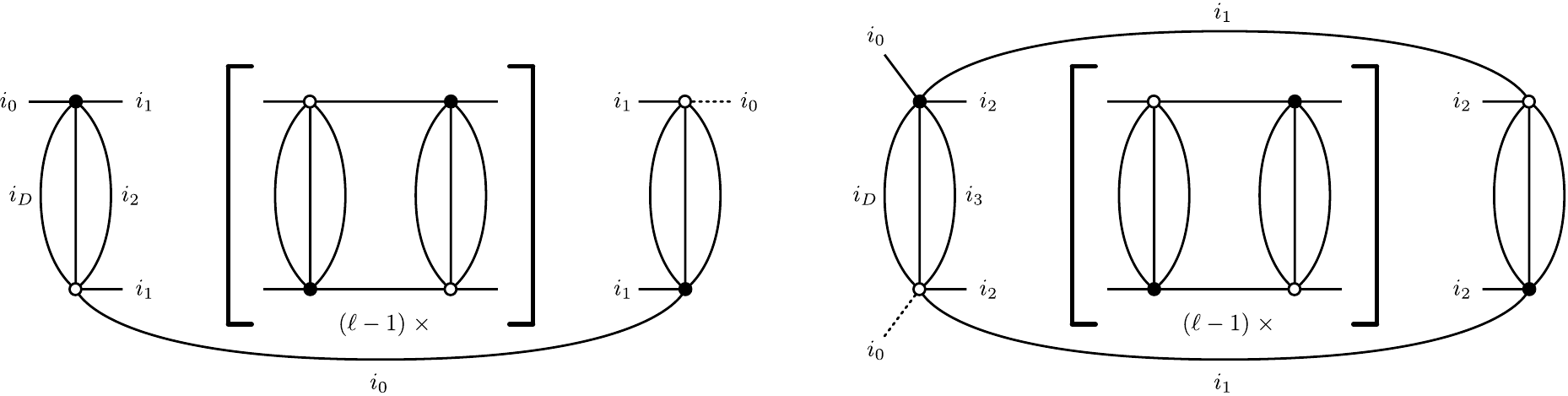}
\caption{\label{fig:el2dipole} The elementary 2--dipole insertions $\mathcal{T}_{\ell,0,0}$.}
\end{figure}

\begin{proposition}
\label{prop:nlographs}
For the \iid\ model, the set of \nlo\ connected 2--point graphs is:
\begin{equation}
T = \bigcup_{\substack{\ell \geq 1\\[0.05cm] m \geq 0\\[0.05cm] n
\geq 0}} T_{\ell,m,n}\;.\nonumber
\end{equation}
\end{proposition}
The subsets $T_{\ell,m,n}$ are defined as follows:
\begin{description}
\item[$\ell$, $m = 0$, $n=0$:]  The graphs in $T_{\ell,0,0}$ are the 2--point graphs obtained by cutting the edges of $\cG_{\twodipole,\ell}$, that is, they are the elementary 2--dipoles illustrated in Figure \ref{fig:el2dipole}.

\item[$\ell$, $m$, $n = 0$:]  The graphs in $T_{\ell,m,0}$ are obtained by replacing an \textit{interior} active edge of an elementary melon with a graph from $T_{\ell,m-1,0}$. Thus, they have the generic form drawn in Figure \ref{fig:nloInsertions}.

 \begin{figure}[htb]
\centering
\includegraphics[scale = 1]{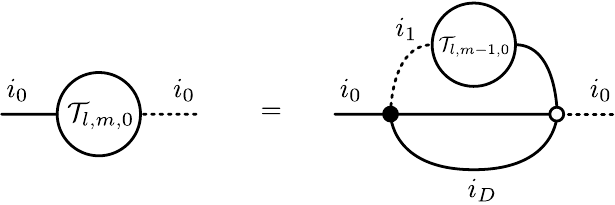}
\caption{\label{fig:nloInsertions} The recursive definition of the elements in $T_{\ell, m,0}$, where $\mathcal{T}_{\ell, m,0} \in T_{\ell, m,0}$ and $\mathcal{T}_{\ell, m-1,0}\in T_{\ell, m-1,0}$.}
\end{figure}

\item[$\ell$, $m$ $n$:] The graphs in $T_{\ell,m,n}$ are obtained from those in $T_{\ell, m,n-1}$ by replacing an active edge with an elementary melon.
\end{description}

Thus, the subset of graphs with $2p$ vertices is:
\begin{equation}
T_{p} = \bigcup_{\substack{1\leq \ell\leq \lfloor p/2\rfloor\\[0.05cm] 0\leq m\leq p - 2\ell}} T_{\ell,m, p -2\ell- m}\;.
\end{equation}

As before, the initial combinatorial factor coming from the Taylor expansion for a graph in $S_{\ell, m,n}$ is  $1/(2\ell +m+n)!$, while the graph is obtained from exactly $(2\ell + m+n)!$ Wick contractions. Thus, the final combinatorial factor is $\sym(\widetilde\cG) = 1$.
This concludes the identification of the graphs contributing to the next-to-leading order.

\noindent\textbf{Remark:}
We stress here the important point that the core graphs $\cG_{\twodipole}$ and thus the whole \nlo\ sector correspond to $D$--dimensional cellular complexes of spherical topology, just like the leading--order graphs.  As a result, the 2--point graphs represent the $D$--dimensional ball.


\section{Critical behavior}
\label{sec:critical}

One now turns to an analysis of $E_{N,g}$, which as one may recall is the free energy at given values of the expansion parameter $N$ and coupling constant $g$.  This may be expanded in both parameters:
\begin{equation}
\label{eq:series}
E_{N,g} = \sum_{\omega}\sum_p E_{\omega,p}\; N^{-\omega} g^p
 \qquad \textrm{where} \qquad E_{\omega, p} = \sum_{\cG\;:\;\substack{\omega(\cG) = \omega \\[0.05cm]  |\cV_\cG| = 2p}} \frac{1}{\textsc{sym}(\cG)}\;.
\end{equation}
The quantity $E_{\omega,p}$ counts the number of closed $(D+1)$--colored graphs with a given degree and a given number of vertices (weighted by the relevant symmetry factors).  In turn, the large--$p$ behaviour of $E_{\omega,p}$ provides the radius of convergence of the series along with the critical exponent:\footnote{For clarity, we have assumed that there are no logarithmic factors contributing to the leading divergence as $g\rightarrow g_{c,\omega}$.  In practice, one should demonstrate this explicitly.}
\begin{equation}
E_{\omega,g} = \sum_p E_{\omega,p}\, g^p \sim A_\omega\; \left(1 - \dfrac{g_{c,\omega}}{g}\right)^{2 - \gamma_\omega} \qquad \textrm{as} \qquad g \rightarrow g_{c,\omega}\;,
\end{equation}
where $A_\omega$ is a constant of proportionality, while $g_{c,\omega}$ and $\gamma_{\omega}$ are the radius of convergence and the susceptibility exponent, respectively. In the coming section, both leading order and next--to--leading order sectors are analysed.  The leading order melonic sector not only provides an invaluable introduction to the techniques used, but also some necessary results. So, it is worth reviewing explicitly, albeit briefly.

As already mentioned, the strategy involves examining the behaviour of the connected 2--point function at the relevant order.  The corresponding behaviour for the free energy can then be found by integrating equation \eqref{eq:2pointSD}.  To succeed, one needs also the 1-particle irreducible (\onepi) 2--point function, given by:
\be
\label{eq:two-point-conserv}
\langle\,\phi^i_{m_i} \, \bar\phi^i_{\bar m_i }\,\rangle_{\onepi} =  {\delta}_{mn}\, \Sigma_{N,g}\; ,
\ee
where ${\Sigma}_{N,g}$ is a constant, depending on $N$ and $g$. There is a convenient identity relating this to the connected 2--point function:
\be
\label{eq:connected-1PI}
{G}_{N,g} = \left(1-{\Sigma}_{N,g}\right)^{-1}\;.
\ee
This leads immediately to the following relations:\footnote{We have used the shorthand:
\begin{equation}
\label{eq:shorthand}
\begin{array}{rcl}
X_{\lo,g} &\equiv& X_{0,g}\\[0.3cm]
X_{\nlo,g} &\equiv& X_{\frac{(D-1)!}{2}(D-2), g}

\end{array}
\end{equation}
where $X$ may be replaced by the suitable observable e.g. $\{E,\,G,\,\Sigma\}$.
}
\be\label{eq:order-by-order}
\ba{rclcl}
{ G}_{\lo,g} &=& (1 - {\Sigma}_{\lo,g})^{-1}\\[0.2cm]

 {G}_{\nlo,g} &=& (1 -  {\Sigma}_{\lo,g})^{-1}\;  {\Sigma}_{\nlo,g}\; (1 -  {\Sigma}_{\lo,g})^{-1}&=&  {G}_{\lo,g}\; {\Sigma}_{\nlo,g}\;  {G}_{\lo,g}

\ea
\ee
 One simply needs to find some more equations to close the system and solve it for the desired function. The equation to be used depend on the detailed combinatorial structure of the graphs at each order.


\subsection{Leading order sector}
\label{ssec:criticallo}

The second equation for the leading order 2--point functions descends directly from the melonic structure of the contributing graphs.  A rooted melonic graph has the generic form given in Figure \ref{fig:melonicgeneric}.

\begin{figure}[htb]
\centering
\includegraphics[scale = 1]{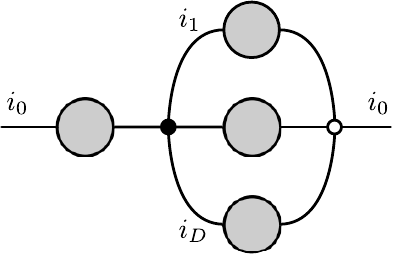}
\caption{\label{fig:melonicgeneric} A generic \lo\ 2--point graph.}
\end{figure}

Thus, any leading order \onepi\ 2--point graph has the form given in  Figure \ref{fig:melonic1PI}, where the shaded circles indicate the insertion of an arbitrary connected rooted melonic graph.

\begin{figure}[htb]
\centering
\includegraphics[scale = 1]{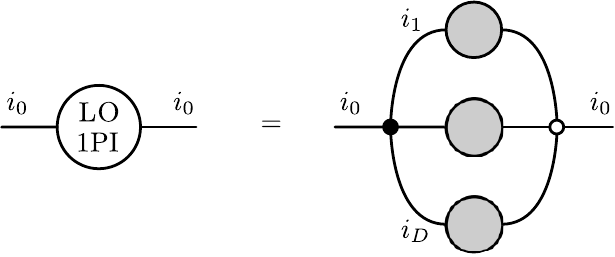}
\caption{\label{fig:melonic1PI} The generic graph contributing to the \lo\ \onepi\ 2--point function. }
\end{figure}

\noindent The diagram illustrates the following mathematical relation:
\be
 {\Sigma}_{\lo,g} = g\,  {G}_{\lo,g}^{D}\;.
\ee
Using \eqref{eq:order-by-order}, this leads to a closed equation for ${G}_{\lo,g}$:
\be\label{eq:closed-melonic}
{G}_{\lo,g} = 1 + g\, {G}_{\lo,g}^{D+1}\;,
\ee
which in turn can be solved via a series expansion for the coefficient:
\be
{G}_{\lo,p}  = C_{p}^{(D+1)} = \frac{1}{(D+1)p + 1} {(D+1)p + 1 \choose p}\;,
\ee
where $C_{p}^{(D+1)}$ are the $(D+1)$--Catalan numbers.  Applying Stirling's formula\footnote{Stirling's formula states that:
\be
\lim_{n\rightarrow \infty} \frac{n!}{\sqrt{2\pi n}\left(\frac{n}{e}\right)^n} = 1
\ee} to the series coefficients, one can examine their large order behavior:
\be\label{eq:melonic-coef-large}
{G}_{\lo,p} \sim  \beta_{\lo}\; (g_{c,\lo})^{-p}\; p^{-\frac{3}{2}}\quad\quad \textrm{where} \quad\quad \left\{
{\renewcommand{\arraystretch}{2}\ba{rcl}
\beta_{\lo} &=& \dfrac{e}{\sqrt{2\pi}} \sqrt{\dfrac{D+1}{D^3}}\\
g_{c,\lo} &=& \dfrac{D^D}{(D+1)^{D+1}}
\ea}\right.
\ee
Such a series has a radius of convergence $g_{c,\lo}$ and the behaviour of the series in the vicinity of $g_{c,\lo}$ is given by:
\be\label{eq:melonic-2p-critical}
{G}_{\lo,g} \sim  \left(1 - \frac{g}{g_{c,\lo}}\right)^{\frac{1}{2}}\;.
\ee
 This implies the following critical behaviour for the free energy:
\be\label{eq:melonic-free-critical}
E_{\lo,g} \sim A_{\lo}  \left(1-\frac{g}{g_{c,\lo}}\right)^{2-\gamma_{\lo}}\quad\quad \textrm{where} \quad\quad \gamma_{\lo} = \frac{1}{2}\;.
\ee
In this context, $\gamma_{\lo}$ is known as the entropy exponent or susceptibility.


\subsection{Next--to--leading order sector}
\label{ssec:criticalnlo}

  The analysis at \nlo\ proceeds similarly. In effect, it requires one to obtain the cardinality of the set $T$ defined in Section \ref{ssec:nlo}.  Consider a \onepi\ 2--point graph occurring at \nlo. It receives contributions from 2--point graphs of the form drawn in Figure \ref{fig:nloequation}. The first corresponds to the insertion an elementary 2--dipole into an elementary melon followed by melonic insertions thereafter (the graphs in $T_{\ell, m,n}$ with $m \geq 1$).  The others correspond to melonic insertions into the elementary 2--dipole graphs (the graphs in $T_{\ell,0,n}$).

\begin{figure}[htb]
\centering
\includegraphics[scale=0.9]{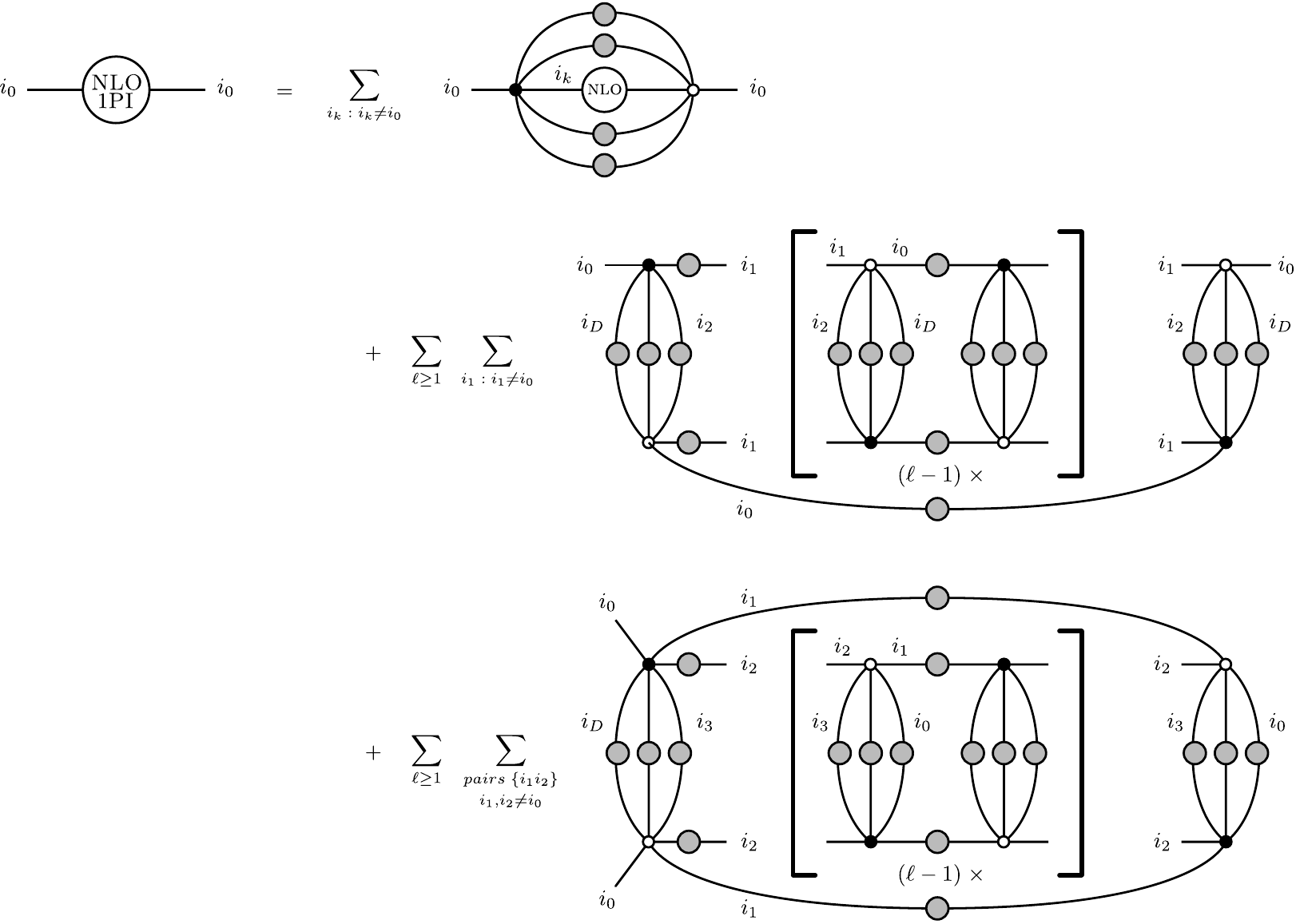}
\caption{\label{fig:nloequation}The \nlo\ consistency equation.}
\end{figure}

In terms of generating functions, this relation translates into the following equation:
\begin{equation}
{\Sigma}_{\nlo,g} = g\, D\, \left[{G}_{\lo,g}\right]^{D-1}\,  {G}_{\nlo,g} + \frac{D(D+1)}{2}\sum_{\ell\geq 1} g^{2\ell}\left[{G}_{\lo,g}\right]^{2\ell(D+1) - 1}  \;,
\end{equation}
where the $D$--dependent factors arise from the choice of species. After rearrangement, one finds:
\begin{equation}
\label{eq:nloSD}
{G}_{\nlo,g} =  \dfrac{\displaystyle \tfrac{D(D+1)}{2} \sum_{\ell\geq 1} g^{2\ell}\,\left[G_{\lo,g}\right]^{2\ell(D+1) +1}  }{1- g\,D\,\left[G_{\lo,g}\right]^{D+1}} = \dfrac{\displaystyle \tfrac{D(D+1)}{2} g^{2}\,\left[G_{\lo,g}\right]^{2D+3}\sum_{\ell\geq 0} \left(G_{\lo,g} - 1\right)^{2\ell}  }{1- g\,D\,\left[G_{\lo,g}\right]^{D+1}}\;,
\end{equation}
where the second equality uses equation \eqref{eq:closed-melonic}. The summation over $\ell$ can be performed explicitly,\footnote{Inverting equation \eqref{eq:closed-melonic} for $g$ as a function of ${G}_{\lo}$, one finds:
\begin{equation}
g\left({G}_{\lo}\right) = (G_{\lo} - 1)/\left[{G}_{\lo}\right]^{D+1}\;.
\end{equation}
One can see that  $g = 0$ corresponds to ${G}_{\lo} = 1$. Moreover,  one knows that the right hand side is a monotonically increasing function of $G_{\lo}$ up to some critical value $G_{\lo,\textrm{critical}}$, where one encounters a stationary point.  By examining its derivatives, one finds that $g\left(G_{\lo}\right)$ has a maximum at ${G}_{\lo,\textrm{critical}} = (D+1)/D$, leading to  $g\left({G}_{\lo,\textrm{critical}}\right) = D^D/(D+1)^{D+1} = g_{c,\lo}$, as expected. Since $|{G}_{\lo,\textrm{critical}} - 1| = 1/D < 1$, we have what we need to resum the series in \eqref{eq:nloSD}.} since $|G_{\lo,g} - 1|< 1$ in the range $0\leq g \leq g_{c,\lo}$ and one gets:
\begin{equation}
{G}_{\nlo,g} = \dfrac{\displaystyle \tfrac{D(D+1)}{2} g^{2}\,\left[G_{\lo,g}\right]^{2D+2} \left({G}_{\lo,g} -2\right)^{-1}}{\left(1- g\,D\,\left[G_{\lo,g}\right]^{D+1}\right)}\;.
\end{equation}
Moveover, upon differentiating \eqref{eq:closed-melonic}, one finds:
\begin{equation}
\label{eq:derivative}
\frac{\partial}{\partial g} G_{\lo,g} = \frac{\left[G_{\lo,g}\right]^{D+2}}{1 - g\,D\,\left[G_{\lo,g}\right]^{D+1}}\;,
\end{equation}
so that:
\begin{equation}
\label{eq:final-equation}
G_{\nlo,g} = g^2\, \frac{D}{2} \left(\frac{1}{G_{\lo,g} - 2}\right)\frac{\partial}{\partial g}\left(\left[G_{\lo,g}\right]^{D+2}\right)\;.
\end{equation}
Knowing the critical behaviour of $G_{\lo,g}$ allows one to determine the behaviour of the \nlo\ series from \eqref{eq:final-equation}:
\begin{equation}
\label{eq:behaveresult}
G_{\nlo,g} \sim  \left(1-\frac{g}{g_{c,\lo}}\right)^{-\frac{1}{2}} \quad\quad \textrm{since} \quad\quad G_{\lo,g} \sim \textrm{constant} + \left(1-\frac{g}{g_{c,\lo}}\right)^{\frac{1}{2}}
\end{equation}
%
This implies the following critical behaviour for the free energy:
\be\label{eq:nlo-critical}
E_{\nlo,g} \sim A_{\nlo}  \left(1-\frac{g}{g_{c,\lo}}\right)^{2-\gamma_{\nlo}}\quad\quad \textrm{where} \quad\quad \gamma_{\nlo} = \frac{3}{2}\;.
\ee

We find then that the critical point of the $\nlo$ is the same as that of the leading order, while the critical exponent differs. This is exactly the property indicating the potential for a double scaling limit.

\section{Discussion: a double scaling limit?}
\label{sec:discussion}

Before concluding, we would like to discuss briefly the implications of our results for the existence and nature of the double scaling limit.
With the above analysis, we have seen that:
\begin{equation}
\label{eq:finalseries}
E_{N,g} \sim A_{\lo}\; N^0 \left(1 - \frac{g}{g_c}\right)^{\frac{3}{2}} + A_{\nlo}\; N^{2 - D} \left(1 - \frac{g}{g_c}\right)^{\frac{1}{2}} + \dots
\end{equation}
Obviously, these two pieces of information already allow one to conjecture a possible form of double scaling limit, whose actual realization relies on the following two properties:
\begin{description}
\item[A1:] There is subset (perhaps with infinity cardinality) of orders, whose elements are labelled by $m\in \mathbb{N}_0$, behaving as:
\begin{equation}
\label{eq:conjecture}
E_{m,g} \sim  A_m\; N^{m(2 - D)} \left(1 - \frac{g}{g_c}\right)^{\frac{3}{2} - m}
\end{equation}
\item[A2:] All other orders are still washed away in the double scaling limit.
\end{description}

If one assumes these two properties, then the double scaling limit of $(N^{3(D-2)/2}E_{N,g})$ ensues from:
\begin{equation}
\label{eq:dsl}
N\rightarrow \infty,\qquad g\rightarrow g_c,\qquad \textrm{such that} \qquad N\left(1 - \frac{g}{g_c}\right)^{1/(D-2)}= \kappa\;.
\end{equation}

We do not have a proof for the above assumptions. The proposed form of double scaling remains, therefore, a conjecture. We can give, however, some supporting arguments for it.

As regards (\textbf{A1}),  there are certainly core graphs weighted by $N^{m(2-D)}$, namely, those that reduce to the supermelon graph through the annihilation of $m$ successive 2--dipoles.\footnote{A subtlety of 2--dipole annihilation is that one may need to insert 1--dipoles briefly in order to effect the annihilation.}

Having said that, one would need to show that the sectors, which we call \textbf{2--dipole sectors}, generated from these core graphs each give rise to a resummable series with the radius of convergence and critical exponent conjectured in \eqref{eq:conjecture}.

Moreover, for a given $m$, there may be other sectors of graphs weighted by $N^{m(2-D)}$.  If such sectors exist, one would need to show the realization of one of the following two scenarios, either \textit{i}) they have the same behaviour as the 2--dipole sectors or \textit{ii}) they do not have the same behaviour, but the radius of convergence is larger then $g_c$, so that they are washed away in the limit.

To prove (\textbf{A2}), one would need to show that the 2--dipole sectors dominate over all other sectors.  This appears a rather daunting task. One might attempt to tackle it by showing that for sectors generated from core graphs with, say, $2p$ vertices, then the 2--dipole sector generated from core graphs with $2p$ vertices dominate.

\begin{figure}[htb]
\centering
\includegraphics[scale = 0.9]{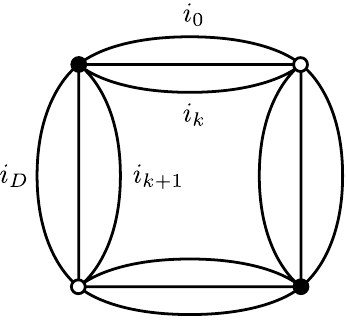}
\caption{\label{fig:k_dipole} A subdominant core graph with four vertices.}
\end{figure}

For example, we can show that this is already for those sectors generated from core graphs with four vertices. As illustrated in Figure \ref{fig:k_dipole}, all such core graphs are obtained from the supermelon core graph after the creation of a $k$--dipole of some species with $k \geq 3$ and denoted here by $\cG_{k\textrm{--dipole}}$.  The techniques that we developed for the \nlo\ sector work on these sectors perfectly.\footnote{In fact, for the sectors generated from $\cG_{k\textrm{--dipole}}$ with $k\geq 3$, the equivalent of the closed graphs $\cG_{\twodipole,\ell}$ with $\ell>1$ do not exist. Thus, the set of closed graphs for these sectors is simply $S = \cup_{n\geq 0} S_{1,n}$ where $S_{1,n} = \{G_{k\textrm{--dipole}}\; \textit{with arbitrary melonic insertions}\}$.  This, however, is already enough to determine the same critical behaviour.}  One finds that while their weight in the $1/N$--expansion is $N^{D + (k - 1)(k - D)}$, they generate a resummable series with behaviour $(1 - g/g_c)^{1/2}$. Thus, their contribution to the series $(N^{3(D-2)/2}E_{N,g})$ is:
\begin{equation}
N^{ 3(D-2)/2 + (k - 1)(k - D)}\left(1 - \frac{g}{g_c}\right)^{\frac{1}{2}}\longrightarrow 0
\end{equation}
in the double scaling limit defined by \eqref{eq:dsl}. The aim is to extend such arguments to higher orders in the expansion.

There are two points of further interest pertaining to the double scaling limit conjectured above.

The first is a comparison to the double scaling limit of the single matrix model.  In that context, the double scaling limit captures contributions from all topologies. For $D = 3$, rather than contain contributions from all topologies, the limit above captures the complete spherical sector of the \iid\ tensor model, since graphs representing the 3--sphere may be obtained from $\mathcal{G}_{\supermelon}$ by some sequence of  1--dipole and 2--dipole moves .  For $D>3$,  the limit captures less and less of the spherical sector, since performing $k$--dipole moves (with $k> 2$) on $\mathcal{G}_{\supermelon}$ may also result in a graph encoding the $D$--sphere. This means that, the higher the dimension, the harder it is for the simplest tensor models to capture  the topology of the continuum spacetime they aim to describe, not to mention its geometry. Richer models involving more coupling constants, thus multiple-scaling limits, or more structured amplitudes, depending on a richer set of data, are then called for.

The second utilizes the gravitational interpretation provided by Euclidean Dynamical Triangulations.  As before, it allows one to enter a regime of finite Newton's constant, by defining a renormalized constant:
\begin{equation}
\frac{1}{G_R} = \frac{8\log \kappa}{\cV_{D-2}} = \frac{1}{G} + \frac{8}{\cV_{D-2}} \log\left(1 - \frac{g}{g_c}\right)^{1/(D-2)}\;.
\end{equation}
However, for $D>2$, Newton's constant is dimensionful. As a result, this $G_R$ is finite in the large--volume limit rather than the continuum limit, where one has also $a\rightarrow 0$ and which will require a more subtle analysis.

\section{Conclusions}
\label{sec:conclusions}

We have studied the next-to-leading order in the large-N expansion of \iid\  tensor models, in any
dimension. We have identified the graphs contributing to it, corresponding 
 to families $S_{\ell,n}$ of graphs.  We then studied the critical behaviour of the free
energy for such graphs. The result is that the critical value of the coupling constant is the same
as at leading order, while the critical (susceptibility) exponent is different. These results
support the possibility of a double scaling limit capturing more properties of the sum over
triangulations for spherical topologies, and suggest the form that this double scaling may take, as
we have discussed in some detail. The analysis we performed for such simple tensor models can also
form the basis of a similar analysis for more involved tensor models as well as for proper TGFTs.

\appendix

\section{Next--to--leading order: technical details}
\label{app:nloTech}

In this appendix, we collect the needed technical definitions, report the details of our results and present the corresponding proofs.

\subsection{Some features of (rooted) melonic graphs}
\label{ssec:melonfeatures}

\subsubsection{Irreducibility}

To begin, we generalize the idea of $n$--particle reducibility to include some color information.
\begin{definition}
\label{def:vecPI}
A graph is said to be $(a_0,\dots,a_D)$--particle irreducible, if it remains connected, having cut $a_j$ edges of color $j$, for \textit{all} $j$.
\end{definition}
We shall denote by $\vec e_{j_1\dots j_k}$ the vector with components $j_1,\dots,j_k$ equalling 0 and with the rest equalling 1.  Here are some easily verified facts about $(D+1)$--colored graphs:
\begin{description}
\item[Q1:] All connected closed $(D+1)$--colored graphs are $\vec e_j$--particle irreducible, for all $j\in\{0,\dots, D\}$.  In other words, one may cut along any $D$ distinctly colored edges and the graph remains connected.

\item[Q2:] All \onepi\ 2--point $(D+1)$--colored graphs with external edges of color $\{i\}$, other than the elementary melon, are $\vec e_j$--particle irreducible, for all $j\in\{\widehat i\}$. The elementary melon is the exception, since it does not have an internal edge of color $\{i\}$. Rather the elementary melon is of species $\{i\}$ is $\vec e_{ij}$--particle irreducible, for all $j\in\{\widehat i\}$.

\end{description}
Now we specialize to melonic graphs.

\subsubsection{Melonic vertex pairs}  Any closed $(D+1)$--colored graph has equal numbers of black and white vertices.  Thus, in principle, there are many possible ways to partition these vertices into black--white pairs.  However, certain graph properties serve to distinguish particular pairings.
\begin{definition}
Rooted melonic graphs have a \textbf{melonic vertex pairing} defined at:
\begin{description}
\label{def:melonicvp}
\item[$p=1$:] An elementary melon has just two vertices and these form a melonic vertex pair.
\item[$p=k$:] Such a graph is constructed via the iterative insertion of elementary melons.  Each elementary melon has a pair of vertices. Thus, the vertices of the graph are paired according to the elementary melon, within which they were inserted.
\end{description}
\end{definition}
As a result, for a given vertex in a rooted melonic graph, one may identify its paired vertex by deleting all sub-nested melons. Moreover, for a closed melonic graph $\cG$, one may construct a melonic vertex pairing by inserting a fictitious cut along one edge and utilizing the resultant rooted graph. From this, one can derive a number of properties:
\begin{description}
\item[P1:\label{itm:unique}] A closed melonic graph $\cG$ has a unique melonic vertex pairing.
\item[P2:] Consider cutting the two edges of some color $\{i\}$ that are incident to the vertices of a melonic vertex pair. This splits the graph $\cG$ into two disjoint connected components (unless the edge of color $\{i\}$ joined the vertices of the vertex pair directly).

\item[P3:] Consider the k--bubbles of $\cG$. These are also melonic and preserve the melonic vertex pairing defined by $\cG$.

\item[P4:] Consider a melonic vertex pair with vertices $v$ and $\bar v$, then $v$ lies in a face of $\cG$ if and only if $\bar v$ lies in that face.  In other words, they lie in the same $D(D+1)/2$ faces.

\end{description}

\subsubsection{1--dipole creation}
\label{sssec:dipolecreate}

Let us remind ourselves that:
\begin{definition}
\label{def:kdipole}
A $k$--dipole $d_k$ is a subset of $\cG$ comprising of two vertices $v$, $\bar v$ such that:
\begin{description}
\item[--] $v$ and $\bar v$ share $k$ edges of colors $i_0,\dots, i_{k-1}$;
\item[--] $v$ and $\bar v$ lie in distinct $(D+1 - k)$--bubbles. In the arguments below, we say that the $k$--dipole \lq\lq separates\rq\rq\ the vertices in the two bubbles.
\end{description}
\end{definition}

\begin{lemma}
\label{lem:separate}
Consider a melonic graph. Then, 1--dipole creation separates some melonic vertex pair.
\end{lemma}
\begin{proof}
At the outset, let us note that closed $(D+1)$--colored graphs are 1--particle irreducible for $D\geq 1$. Then, insert a 1--dipole of species $\{j\}$ and suppose that it does not separate any melonic vertex pair.  Thus, the two new vertices, $v$ and $\bar v$, form a melonic vertex pair themselves. Therefore, by (\textbf{P2}), cutting the edges of color $\{i\}$, with $i\in\{\widehat j\}$, emanating from $v$ and $\bar v$ disconnects the graph into two connected components.  But one of these cuts lies in each of the $D$--bubbles of species $\{\widehat j\}$ (separated by the 1--dipole).  Thus, the  $D$--bubble is disconnected by a single cut, which contradicts the fact that it is \onepi.
\end{proof}

Now, consider a melonic vertex pair $v\bar v$. By property (\textbf{P4}), they both lie in the same $D(D+1)/2$ faces, one of each species $\{j,k\}$.  Now, let us examine these faces more closely, and in particular, their bounding edges:
\begin{description}
\item[--] $\cE^{v\bar v}_{j,k}$ is the set of edges of color $\{j\}$ that lie in the boundary of the face of species $\{j,k\}$.
\item[--] $\cE^{v\bar v}_{j} = \bigcap_{k\in \{\widehat j\}} \cE^{v\bar v}_{j,k}$ is the set of edges of color $\{j\}$ that lie in the boundary of \textit{all} such faces.
\end{description}

\begin{definition}
\label{def:cuts}
Consider a generic (not necessarily melonic) graph that possesses a melonic vertex pairing. $\cE^{v\bar v}_{0\textrm{--dipole}}$ denotes all the edges that may be cut to create a 0--dipole, which separates the melonic vertex pair $v\bar v$. $\cE^{v\bar v}_{1\textrm{--dipole},i}$ denotes all the edges that may be cut to create a 1--dipole of species $\{i\}$, which separates the melonic vertex pair $v\bar v$.
\end{definition}

\begin{lemma}
\label{lem:cutedges}
For a melonic graph:
\begin{equation}
\label{eq:onedipedges}
\cE^{v\bar v}_{0\textrm{--dipole}} = \bigcup_{j\in\{0,\dots, D\}} \cE^{v\bar v}_{j}\quad\quad,\quad\quad
\cE^{v\bar v}_{1\textrm{--dipole}, i} = \bigcup_{j\in\{\widehat i\}} \cE^{v\bar v}_{j}\;.
\end{equation}
\end{lemma}
\begin{proof}

\begin{figure}[htb]
\centering
\includegraphics[scale = 1.3]{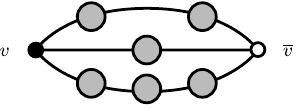}
\caption{\label{fig:melons} A closed melonic graph, with melonic vertex pair $v\bar v$ and \onepi\ melonic decorations.}
\end{figure}

Consider a closed melonic graph in Figure \ref{fig:melons}, with melonic vertex pair $v\bar v$.  We have decomposed all melonic decorations into their \onepi\ components.
$\cE^{v\bar v}_{0\textrm{--dipole}}$ contains all, and only, those edges exterior to all \onepi\ decorations.   That  $\cE^{v\bar v}_{0\textrm{--dipole}}$ cannot contain edges interior to some \onepi decoration follows very simply from the irreducibility property (\textbf{Q2}).

Similarly, $\cE^{v\bar v}_{1\textrm{--dipole}, i}$ contains those edges of color $\{j\}$ (with $j\in\{\widehat i\}$) exterior to all \onepi\ decorations.
\end{proof}

Finally, a result from \cite{critical}:
\begin{proposition}
\label{prop:melonclosure}
The set of rooted melonic graphs $M$ is closed under 1--dipole creation and annihilation.
\end{proposition}

\subsection{Identifying {\sc nlo} core graphs}
\label{ssec:identifying}

At first glance, it is perhaps unsurprising that the graphs illustrated in Figure \ref{fig:core-graphs-nlo}, being also the next simplest in terms of number of vertices, are the \nlo\ core graphs. However, it emerges that it requires quite some care to prove this definitively.   Closer inspection reveals that the graphs have degree:
\be
\omega(\cG_{\twodipole}) = \frac{(D-1)!}{2} (D-2)\;.
\ee
Since the degree of $\cG_{\twodipole}$ scales factorially with $D$, it is unclear at the outset that no other graph, perhaps less superficially obvious, sneaks in with a smaller degree to take the role  of the \nlo\ core graph.   As it stands, the constraints on core graphs are not strong enough to rule out this possibility; one has only that they must be connected with $p > 1$. The aim of this subsection is to develop constraints that rule out this possibility.

\begin{lemma}
\label{lem:2-melonic}
Consider a $(D+1)$--colored {\bf core} graph $\cG$, at order $p$ and $D\geq 4$, with the following properties:
\begin{itemize}
\item[--] $\cG$ possesses exactly two melonic $D$--bubbles, say of species $\{i_0\}$ and $\{i_1\}$.
\item[--] $\cG$ possesses a planar jacket.
\end{itemize}
Then, $\cG$ is the supermelon graph with $p-1$ 2--dipoles of species $\{i_0,i_1\}$ inserted.
\end{lemma}
\begin{proof}
Since $\cG$ is a core graph, it contains a single $D$-bubble for each color, labelled $\cR_{(\widehat i)}$ for $i\in \{0,1,\dots, D\}$.  We shall denote the two melonic $D$-bubbles by $\cR_{(\widehat i_0)}$, $\cR_{(\widehat i_1)}$, respectively. Moreover, these $D$-bubbles each contain all the vertices of $\cG$.

Since $\cR_{(\widehat i_0)}$ and $\cR_{(\widehat i_1)}$ are closed melonic graphs, by (\textbf{P1}), they each have unique melonic vertex pairing. Thus, a priori, a given vertex in $\cG$ has two melonic vertex pairings: one induced by $\cR_{(\widehat i_0)}$ and another by $\cR_{(\widehat i_1)}$.  However, given the property (\textbf{P3}) of melonic vertex pairs, these vertex pairings coincide since both $\cR_{(\widehat i_0)}$ and $\cR_{(\widehat i_1)}$ contain the $(D-1)$--bubbles of species $(\widehat i_0\widehat i_1)$.
Thus, even though $\cG$ is not melonic it has a melonic vertex pairing.

Moreover, since $\cG$ possesses a planar jacket, we can draw it on the plane without crossings, such that the  colored edges incident to each and every vertex are ordered according to some $(D+1)$--cycle $\tau$.

We utilize the planar illustration of $\cG$ in Figure \ref{fig:dipole-k}, where the cycle has the form $\tau = (\dots jkl \dots)$. Consider a melonic vertex pair in $\cG$ and say that they do not share an edge of color $k\in\{\widehat i_0,\widehat i_1\}$.  Our argument has three subcases:
\begin{description}
\item[{\it i})] Neither $j$ nor $l$ equal $i_0$. In $\cR_{(\widehat i_0)}$, property (\textbf{P2}) ensures that cutting the edges of color $k$ disconnects a subgraph from the rest of $\cR_{(\widehat i_0)}$.  Thus, there is a 1-dipole of species $k$ in $\cR_{(\widehat i_0)}$.  Upon reinserting the lines of color $i_1$, the planarity of the jacket ensures that this 1--dipole of species $k$ persists into $\cG$.

\begin{figure}[H]
\centering
\includegraphics[scale = 1]{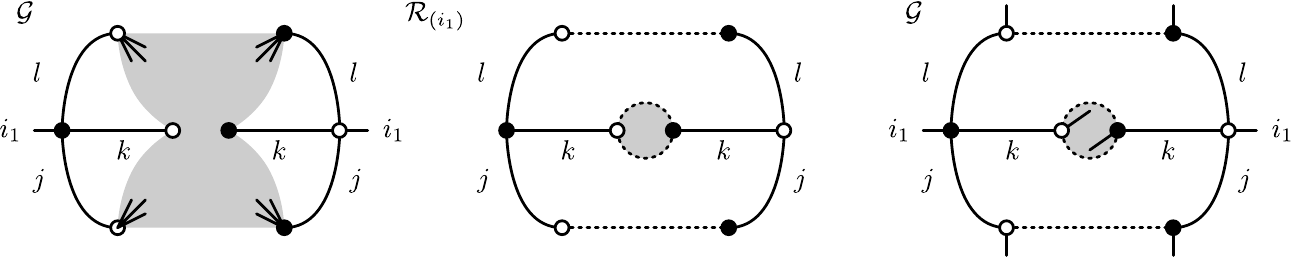}
\caption{\label{fig:dipole-k}Killing 1-dipoles of species $k$.}
\end{figure}

\item[{\it ii})] Neither $j$ nor $l$ equal $i_1$. The argument here is similar to the above with $i_1$ swapped for $i_0$.

\item[{\it iii})] $\{j,l\}=\{i_0,i_1\}$.  In $\cR_{(\widehat i_0)}$, property (\textbf{P2}) ensures that cutting the edges of color $k$ disconnects a subgraph from the rest of $\cR_{(\widehat i_0)}$. As a result, the two distinguished edges of color $k$ lie in the same face of color $(i_1k)$.   Swapping the roles of $i_0$ and $i_1$, we can show that the two distinguished edges of color $k$ lie in the same face of color $(i_0k)$.  Thus, back in $\cG$, there is a 1--dipole of species $k$.

\end{description}

Thus, all lines of color $k\notin\{i_0,i_1\}$ directly connect melonic vertex pairs. In turn, this means that $\cR_{(\widehat i_0)}$ is the $D$--colored supermelon with $p-1$ 1-dipoles of species $i_2$ inserted (and vice versa).  We illustrate $\cR_{(\widehat i_0)}$ in Figure \ref{fig:dipole-i}.

\begin{figure}[H]
\centering
\includegraphics[scale = 0.8]{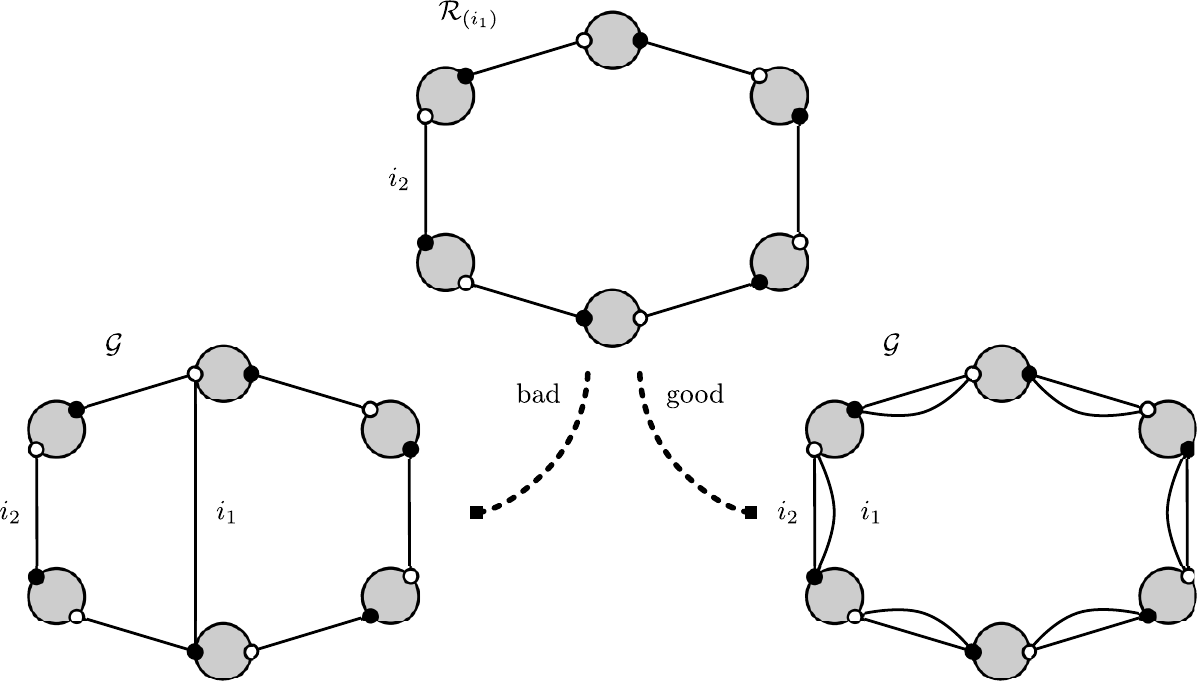}
\caption{\label{fig:dipole-i}Distinguishing good from bad choices for the reinsertion of edges of color $i_1$.}
\end{figure}

Once again, the planarity of the jacket ensures that, upon reinserting the edges of color $i_0$, $\cG$ is the $(D+1)$--colored supermelon with $p-1$ 2-dipoles of species $\{i_0,i_1\}$ inserted.

\end{proof}

A corollary of this statement is the following:
\begin{corollary}
\label{cor:3-melonic}
Consider a $(D+1)$-colored core graph $\cG$ with $D\geq 4$. If $\cG$ possesses three (or more) melonic $D$--bubbles and a planar jacket, then it is the supermelon.
\end{corollary}
Now for the main result of this subsection, the identification of the \nlo\ core graphs:
\begin{repproposition}{prop:core}
Consider the \iid\ tensor model with $D\geq 3$. The graphs $\cG_{\twodipole}$ are the \nlo\ core graphs.
\end{repproposition}
\begin{proof} One uses an inductive argument on the number of colors. A $k$--colored core graph is denoted by $ \up{k}\cG$. To start off, one shows that the statement holds true for $D=3$. While Lemma \ref{lem:2-melonic} does not apply to this case, luckily it is simple to show directly. The following three statements hold in $D=3$:  \textit{i}) $\omega(\up{3}\cG_{\twodipole}) = 1$, so it is certainly a \nlo\ core graph (it is only sub-dominant by one power of $N$); \textit{ii}) $\up{3}\cG_{\twodipole}$ is the only core graph with $p = 2$; \textit{iii}) equation \eqref{eq:degreeDbubble} implies that $\omega(\up{3}\cG) \geq p - 1$, so given the second point, any other graph has a higher degree than $\up{3}\cG_{\twodipole}$. Thus, the \nlo\ core graphs are as proposed: $\up{3}\cG_{\nlo} = \up{3}\cG_{\twodipole}$.

Say next that the statement holds true for the $D$-colored model: $\up{D}\cG_{\nlo} = \up{D}\cG_{\twodipole}$. Thus $\omega(\up{D}\cG_{\nlo})=\frac{(D-2)!}{2} (D-3)$.

Now, say that despite this assumption, the statement does not hold for the $(D+1)$--colored model. Alas, the graphs $\up{D+1}\cG_{\twodipole}$, are trumped by some other graph $\up{D+1}\cG_{\textsc{sneaky}}$. In other words, their degrees satisfy the following inequality:
\be
\label{eq:inequality}
\omega\big(\up{D+1}\cG_{\textsc{sneaky}}\big)\leq \omega\big(\up{D+1}\cG_{\twodipole}\big) = \frac{(D-1)!}{2}(D-2)\,.
\ee
 A generic core graph of the $(D+1)$--colored model satisfies the equality:
\begin{equation}
\label{eq:degree}
\omega\big(\up{D+1}\cG\big) = \frac{(D-1)!}{2} (p-1) + \sum_i \omega\big(\up{D}\cR_{(i)}\big) \,.
\end{equation}
where $\up{D}\cR_{(i)}$ are its $D$--bubbles. In particular, $\up{D+1}\cG_{\textsc{sneaky}}$ obeys such a relation. They key now is to put a lower bound on both the number of vertices and the degrees of the $D$--bubbles and force a contradiction.

First of all, one knows also that $p\geq 2$, since the only core graph with $p=1$ is the supermelon. Secondly, $\up{D+1}\cG_{\textsc{sneaky}}$ possesses at least one planar jacket, else $\omega\big(\up{D+1}\cG_{\textsc{sneaky}}\big) = \sum_{\cJ}g_{\cJ} \geq D!/2$.   Moreover, appealing to Corollary \ref{cor:3-melonic}, one can exclude all cases with three or more melonic D--bubbles.  Thus, given that there are $D+1$ $D$--bubbles in total, the configuration minimizing \eqref{eq:degree} has $D-1$ of them at \nlo.   But remember that $D$--bubbles are $D$--colored graphs, so the lower bound for the degree of $\up{D+1}\cG_{\textsc{sneaky}}$ is already:
\begin{equation}
{\renewcommand{\arraystretch}{2}
\begin{array}{rcl}
\omega(\up{D+1}\cG_{\textsc{sneaky}}) &\geq& \dfrac{(D-1)!}{2}(p-1) + (D-1)\dfrac{(D-2)!}{2}(D-3) + \omega(\up{D}\cR_{(i_0)}) + \omega(\up{D}\cR_{(i_1)}) \\
&= & \dfrac{(D-1)!}{2}(p+D-4) +  \omega(\up{D}\cR_{(i_0)}) + \omega(\up{D}\cR_{(i_1)})\,.
\end{array}
}
\end{equation}
In order to satisfy the bound \eqref{eq:inequality}, one must set $p=2$ and $\omega(\up{D}\cR_{(i_0)}) = \omega(\up{D}\cR_{(i_1)}) = 0$.  Thus, one falls into the case dealt with by Lemma \ref{lem:2-melonic} and a contradiction follows thereafter.
\end{proof}

\subsection{Generating all {\sc nlo} graphs}
\label{ssec:labeling}

From these core graphs, one can generate all graphs at \nlo\ by performing arbitrary sequences of 1--dipole moves. At \lo, such sequences resulted in the insertion of some set of elementary melons into the core graph. At \nlo, while it is certainly true that most 1--dipole moves still result in the creation/annihilation of elementary melons, there is a subset that departs from this rule.  %
The set of graphs $S$ is certainly obtained by performing some sequence of 1-dipoles moves on the \nlo\ core graph $\cG_{\twodipole}$. Our aim is to show now that these are all the graphs.

%
%
%
%

With these properties at our disposal, we can proceed to analyse the \nlo\ sector.

\begin{repdefinition}{def:closedgraphs}
One denotes by $S_{\ell,n}$, the set of closed graphs derived from $\cG_{\twodipole,\ell}$ by inserting arbitrary combinations of $n$ elementary melons.
\end{repdefinition}

\begin{repproposition}{prop:nlographsClosed}
The set of graphs:
\begin{equation}
S = \bigcup_{\substack{\ell \geq 1\\[0.05cm]n\geq 0}} S_{\ell, n}\;.\nonumber
\end{equation}
is closed under 1--dipole creation and annihilation.
\end{repproposition}
\begin{proof}
Every element of $S_{\ell,n}$ is built from the graph $\cG_{\twodipole,\ell}$ decorated by some combination of $n$ melons. A generic graph in $S_{\ell, n}$ is drawn in Figure \ref{fig:closedNLO}. 

\begin{figure}[htb]\centering
\begin{minipage}{0.5\linewidth}
\centering
\includegraphics[scale = 0.9]{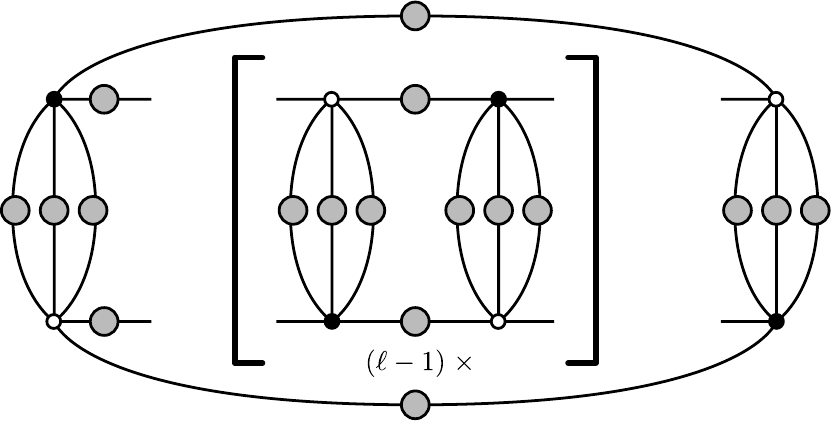}
\caption{\label{fig:closedNLO} A generic element of $S_{\ell, n}$, where $n$ counts the number of elementary melon insertions.}
\end{minipage}\hspace{0.05\linewidth}
\begin{minipage}{0.4\linewidth}
\centering
\includegraphics[scale = 1]{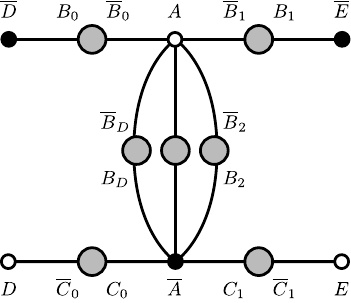}
\caption{\label{fig:nlosubgraph} A partial subgraph.}
\end{minipage}
\end{figure}

\noindent{\bf 1--dipole creation:} From Definition \ref{def:kdipole}, to create a 1--dipole of some species $\{i\}$, one must pick $D$ edges of distinct colors taken from the set $\{\widehat i\}$ that ensure the separation property.   It will emerge that any 1--dipole involves $D$ edges from a single partial subgraph of the type in Figure \ref{fig:nlosubgraph}.

To begin, the analysis of Section \ref{ssec:identifying} shows that the graphs in $S_{\ell,n}$ have
 two melonic $D$--bubbles (of species $\{\widehat i_0\}$ and $\{\widehat i_1\}$), which induce a
unique melonic vertex pairing.\footnote{\label{fn:threed} For $D = 3$, there is little subtlety in that the \nlo\
graphs have four melonic 3--bubbles. For all graphs in $S_{\ell, n}$ with $\ell > 1$, this poses no
problem. However, examining the core graph $\cG_{\twodipole}$ explicitly, one notices that it is
symmetric and does not have a distinguished melonic vertex pairing. For this subset of graphs
there are two different choices of pairing. This subtlety will be important later.}  One can easily
identify this melonic vertex pairing in Figures \ref{fig:closedNLO} and \ref{fig:nlosubgraph}.
While the rooted melonic insertions have their melonic vertex pairs, those vertices exterior to any
melonic decoration are in pairs of type $A\overline A$, $D\overline D$ and so forth.

In general, the creation of a 1--dipole in the graph corresponds to the creation of a 0-- or 1--dipole its melonic $D$--bubbles. By Lemma \ref{lem:separate}, a 1--dipole insertion separates some melonic vertex pair of a melonic $D$--bubble (and a 0--dipole certainly does).  Thus, a 1--dipole insertion separates some melonic vertex pair of the graph.

Recall from Definition \ref{def:cuts}: $\cE^{v\overline v}_{1\textrm{--dipole}, i}$ denotes the set of edges that may be cut to create a 1--dipole of species $\{i\}$ in the graph, so as to separate the melonic vertex pair $v\bar v$.

%

One proceeds on a case by case basis.
\begin{description}
\item[Case 1:]  Say one wishes to create 1--dipole of species $\{i\}$ to separate a melonic vertex pair of type $A\overline A$.

The key is to examine the effect of the corresponding 0--/1--dipole on the melonic $D$--bubbles. This restricts the elements of $\cE^{v\overline v}_{1\textrm{--dipole}, i}$.  One finds that in all cases, $\cE^{v\overline v}_{1\textrm{--dipole}, i}$ contains those edges of species $j$ (with $j\in \{\widehat i\}$) in Figure \ref{fig:nlosubgraph} that are exterior to any melonic decoration.  Any choice of $D$ distinctly colored edges from this set sends the graph in $S_{\ell, n}$ to a graph in $S_{\ell, n+1}$.

Without loss of generality, let us assume that the melonic decoration between $B_j$ and $\overline B_j$ is \onepi, so that $B_j$ and $\overline B_j$ form a melonic vertex pair. If it is not, then one decomposes this decoration into its \onepi\ components and proceeds with one of the resulting vertex pairs.

\item[Case 2:]  Say one wishes to create a 1--dipole of species $\{i\}$ to separate a melonic vertex pair of type $B_j\overline B_j$.

\begin{description}
\item[$j\in\{\widehat i_0,\widehat i_1\}$:]  Unless $D = 3$ and $\ell = 1$, 
$\cE^{B_{j}\overline B_{j}}_{1\textrm{--dipole}, i}$ contains only edges interior to the rooted
melonic subgraph $A\overline B_j B_j \overline A$. Thus, we may apply Proposition
\ref{prop:melonclosure} and one sees that 1--dipole creation sends the graph from $S_{\ell,n}$ to a
graph in $S_{\ell, n+1}$.
The case $D=3$ and $\ell=1$ is special, as explained in footnote \ref{fn:threed}, since there are two possible pairings for the four vertices exterior to any melonic decoration. However, a pairing can always be chosen so that the assumption $j\in \{ i_0,  i_1\}$ is valid and therefore we can avoid this case.

\item[$j = i_0$, $i\in\{\widehat i_1\}$:]  $\cE^{B_{0}\overline B_{0}}_{1\textrm{--dipole}, i}$ contains only edges interior to the rooted melonic subgraph  $A\overline B_0 B_0 \overline D$.  As above, we may apply Proposition \ref{prop:melonclosure}.

\item[$j = i_1$, $i\in\{\widehat i_0\}$:] As above, with the roles of $i_0$ and $i_1$ swapped.

\item[$j = i_0$, $i = i_1$:]  As elements of $\cE^{B_{0}\overline B_{0}}_{1\textrm{--dipole},
i_1}$, those edges of color $\{k\}$ with $k\in\{\widehat i_0\}$ are interior to the melonic
decoration between $B_0$ and $\overline B_0$.  Those edges of color $\{i_0\}$ lie along the paths $A\overline B_0 B_0 \overline D$ and $D\overline C_0 C_0 \overline A$. If the melonic decorations  $B_0\overline B_0$ and $C_0\overline C_0$ are \onepi, then the viable edges of color $\{i_0\}$ are explicitly: $A\overline B_0$, $B_0\overline D$, $D\overline
C_0$, $C_0\overline A$. If these melonic decorations are not \onepi, then the edges of color $\{i_0\}$ separating the \onepi\ components must be added to the viable set. If one chooses to cut an edge of color $\{i_0\}$ contained in the path $A\overline B_0 B_0 \overline D$, then one sends the graph in $S_{\ell, n}$ to a graph in
$S_{\ell, n+1}$. However, if one chooses an edge of color $\{i_0\}$ contained in the path $D\overline C_0 C_0 \overline A$, then one sends
the graph in $S_{\ell, n}$ to a graph in $S_{\ell+1, n-1}$. This is the type of move illustrated in
Fig \ref{fig:extension}.

\item[$j = i_1$, $i = i_0$:] As above, with the roles of $i_0$ and $i_1$ swapped.

\end{description}

\item[Case 3:] A 1--dipole to separate a melonic vertex pair nested at least once inside a melonic decoration. In this case, the viable edge set is contained entirely within the melonic insertion.  Thus, we may apply Proposition \ref{prop:melonclosure} yet again.

\end{description}

\noindent{\bf 1--dipole annihilation:} This involves picking an edge of color $\{i\}$ such that after contraction, the number of $D$--bubbles of species $\{\widehat i\}$ is reduced by one.
\begin{description}
\item[Case 1:] One picks a 1--dipole edge \textit{interior} to some melonic insertion.  Then, Proposition \ref{prop:melonclosure} holds and the graph in $S_{\ell, n}$ is sent to a graph in $S_{\ell, n-1}$.
\item[Case 2:] One picks a 1--dipole edge \textit{exterior} to any melonic insertion. Then, the only 1--dipole edges are those or color $\{i_0\}$ and $\{i_1\}$. Consider the path $A\rightarrow \overline B_0\rightarrow  B_0 \rightarrow  \overline D$.  There are two subcases:

\begin{description}
\item[--] The path from $A$ to $\overline D$ in Figure \ref{fig:nlosubgraph} is decorated by a melonic insertion. Then, the edge of color $\{i_0\}$ joining vertices $A$ and $\overline B_0$ is a 1--dipole edge.  Contracting this 1--dipole sends the graph in $S_{\ell,n}$ to a graph in $S_{\ell, n-1}$.
\item[--] The path from $A$ to $\overline D$ is undecorated.  Then, the edge joining the vertices $A$ and $\widetilde D$ is a 1--dipole edge, whose contraction sends the graph in $S_{\ell, n}$ to a graph in $S_{\ell -1, n+1}$.  This is the inverse of the move illustrated in Figure \ref{fig:extension}.

\end{description}
\end{description}

\end{proof}
By cutting some single edge of the graphs in $S$, one generates the set of \nlo\ connected 2--point graphs, denoted by $T$. These are defined and catalogued in Section \ref{ssec:nlo}.
\begin{repproposition}{prop:nlographs}
For the \iid\ model, the set of \nlo\ connected 2--point graphs is:
\begin{equation}
T = \bigcup_{\substack{\ell\geq 1 \\[0.05cm] m \geq 0\\[0.05cm] n \geq 0}} T_{\ell, m,n}\;.\nonumber
\end{equation}
\end{repproposition}
\begin{proof}
This is a proof by inspection.  One examines the graph in $S_{\ell, n}$, as drawn in Figure \ref{fig:closedNLO}.  For our purposes, there are two types edges in such a graph:
\begin{description}
\item[--] Cutting an edge interior to some melonic insertion sends a graph in $S_{\ell,n}$ to some graph in $T_{\ell,m,n-m}$, where $m$ denotes the level of nesting (within the melonic insertion) of the cut edge;
\item[--] Cutting an edge exterior to any melonic insertion sends a graph in $S_{\ell,n}$ to some graph in $T_{\ell,0,n}$.
\end{description}
Moreover, taking a graph any graph in $T_{\ell, m,n}$ and joining its two open edges produces a graph in $S_{\ell,n+m}$.
\end{proof}

\section*{Acknowledgements}
We thank Vincent Rivasseau, Razvan Gurau and especially Dario Benedetti for frequent helpful discussions.  Wojciech Kami\'nski acknowledges grant of Polish Narodowe Centrum Nauki number 501/11-02-00/66-4162.

{\footnotesize

}


\begin{thebibliography}{99}




\bibitem{GFT1}
  D.~Oriti,
  {\it The Group field theory approach to quantum gravity},
  in  {\sl Approaches to Quantum Gravity}, D. Oriti, ed., Cambridge
Cambridge University Press, Cambridge (2009), [arXiv: gr-qc/0607032]


\bibitem{GFT2}
  D.~Oriti,
  ``Quantum gravity as a quantum field theory of simplicial geometry,''
  In *Fauser, B. (ed.) et al.: Quantum gravity* 101-126
  [gr-qc/0512103].



\bibitem{GFT3} D. Oriti,
{\sl The microscopic dynamics of quantum space as a group field theory},
in  {\sl Foundations of space and time}, G. Ellis, J. Murugan, A. Weltman (eds), Cambridge University Press, Cambridge (2012), arXiv:1110.5606 [hep-th]

\bibitem{vincent}
  V.~Rivasseau,
{\it Quantum Gravity and Renormalization: The Tensor Track},
  AIP Conf.\ Proc.\  {\bf 1444}, 18 (2011)
  [arXiv:1112.5104 [hep-th]];
  V. Rivasseau, {\it The Tensor Track: an Update},
  arXiv:1209.5284 [hep-th]

\bibitem{matrix}
 P.~Di Francesco, P.~H.~Ginsparg, J.~Zinn-Justin,
  {\it``2-D Gravity and random matrices,''}
  Phys.\ Rept.\  {\bf 254 } (1995)  1-133.




\bibitem{review}
R.~Gurau and J.~P.~Ryan,
  {\it Colored Tensor Models - a review,}
  SIGMA {\bf 8} (2012) 020
  [arXiv:1109.4812 [hep-th]].

\bibitem{tensor}
  N.~Sasakura,
  {\it``Tensor model for gravity and orientability of manifold,''}
  Mod.\ Phys.\ Lett.\  A {\bf 6}, 2613 (1991);
  J.~Ambjorn, B.~Durhuus and T.~Jonsson,
  {\it``Three-Dimensional Simplicial Quantum Gravity And Generalized Matrix
  Models,''}
  Mod.\ Phys.\ Lett.\  A {\bf 6}, 1133 (1991).

















 \bibitem{AristideDaniele}
A.~Baratin and D.~Oriti,
  {\it Group field theory and simplicial gravity path integrals: A model for Holst-Plebanski gravity,}
  Phys.\ Rev.\ D {\bf 85} (2012) 044003
  [arXiv:1111.5842 [hep-th]].






\bibitem{lqg}
C.\ Rovelli. {\it Quantum Gravity,} (Cambridge University Press, Cambridge 2004.)

T.\ Thiemann, {\it Modern canonical quantum general relativity,} (Cambridge University Press, Cambridge 2007.)

A.\ Ashtekar and J.\ Lewandowski, {\it Background independent quantum gravity: A status report,} Class. Quant. Grav. 21, R53 (2004) [arXiv:gr-qc/0404018].

 \bibitem{DT}
J.~Ambjorn, B.~Durhuus and T.~Jonsson,
  {\it Quantum geometry. A statistical field theory approach,}
  Cambridge, UK: Univ. Pr., 1997. (Cambridge Monographs in Mathematical Physics). 363pp;
J.~Ambjorn, J.~Jurkiewicz and R.~Loll,
  {\it Dynamically triangulating Lorentzian quantum gravity,}
  Nucl.\ Phys.\ B {\bf 610} (2001) 347
  [hep-th/0105267].


\bibitem{dupuis}
M.~Dupuis, J.~P.~Ryan and S.~Speziale,
  {\it Discrete gravity models and Loop Quantum Gravity: a short review,}
  SIGMA {\bf 8} (2012) 052
  [arXiv:1204.5394 [gr-qc]].

  \bibitem{EPRL}
J.~Engle, E.~Livine, R.~Pereira and C.~Rovelli,
  {\it LQG vertex with finite Immirzi parameter,}
  Nucl.\ Phys.\ B {\bf 799} (2008) 136
  [arXiv:0711.0146 [gr-qc]];
  L.~Freidel and K.~Krasnov,
  {\it A New Spin Foam Model for 4d Gravity,}
  Class.\ Quant.\ Grav.\  {\bf 25} (2008) 125018
  [arXiv:0708.1595 [gr-qc]].



  \bibitem{GFT-EPRL}
  J.~Ben Geloun, R.~Gurau and V.~Rivasseau,
  {\it ``EPRL/FK Group Field Theory,''}
  Europhys.\ Lett.\  {\bf 92}, 60008 (2010)
  [arXiv:1008.0354 [hep-th]].

\bibitem{EteraMaite} M. Dupuis, E. Livine,
{\it Holomorphic Simplicity Constraints for 4d Spinfoam Models}, Class.Quant.Grav. 28 (2011) 215022, arXiv:1104.3683 [gr-qc]


\bibitem{sf5}
  A.~Riello,
 {\it ``Self-Energy in the Lorentzian ERPL-FK Spin Foam Model of Quantum Gravity,''}
  arXiv:1302.1781 [gr-qc].

\bibitem{SF} A. Perez, The Spin Foam Approach to Quantum Gravity, Living Rev. Relativity 16,  (2013), arXiv:1205.2019;
  C.~Rovelli,
  ``Zakopane lectures on loop gravity,''
  PoS QGQGS {\bf 2011}, 003 (2011)
  [arXiv:1102.3660 [gr-qc]].



\bibitem{josephvalentin}
  J.~Ben Geloun and V.~Bonzom,
  {\it Radiative corrections in the Boulatov-Ooguri tensor model: the 2-point function},
  Int.\ J.\ Theor.\ Phys.\  {\bf 50}, 2819 (2011)
  [arXiv:1101.4294 [hep-th]]


\bibitem{expansion1}
R.~Gurau,
  {\it The 1/N expansion of colored tensor models,}
  Annales Henri Poincare {\bf 12} (2011) 829
  [arXiv:1011.2726 [gr-qc]].

   \bibitem{expansion2}
   R.~Gurau and V.~Rivasseau,
  {\it The 1/N expansion of colored tensor models in arbitrary dimension,}
  Europhys.\ Lett.\  {\bf 95}, 50004 (2011)
  [arXiv:1101.4182 [gr-qc]].

 \bibitem{expansion3}
R.~Gurau,
  {\it The complete 1/N expansion of colored tensor models in arbitrary dimension,}
  Annales Henri Poincare {\bf 13} (2012) 399
  [arXiv:1102.5759 [gr-qc]].

  \bibitem{expansion4}
V.~Bonzom,
  {\it New 1/N expansions in random tensor models,}
  arXiv:1211.1657 [hep-th].

 \bibitem{expansion5}
  S.~Dartois, V.~Rivasseau and A.~Tanasa,
  {\it The 1/N expansion of multi-orientable random tensor models,}
  arXiv:1301.1535 [hep-th].

\bibitem{critical}
V.~Bonzom, R.~Gurau, A.~Riello and V.~Rivasseau,
  {\it Critical behavior of colored tensor models in the large N limit,}
  Nucl.\ Phys.\ B {\bf 853} (2011) 174
  [arXiv:1105.3122 [hep-th]].


\bibitem{spinglasses}
  V.~Bonzom, R.~Gurau and M. Smerlak, {\it Universality in p-spin glasses with correlated disorder},
  arXiv:1206.5539 [hep-th]

\bibitem{ising}
V.~Bonzom, R.~Gurau and V.~Rivasseau,
  {\it The Ising Model on Random Lattices in Arbitrary Dimensions,}
  Phys.\ Lett.\ B {\bf 711} (2012) 88
  [arXiv:1108.6269 [hep-th]].
 


\bibitem{dimers}
 V.~Bonzom,
  {\it Multicritical tensor models and hard dimers on spherical random lattices,}
  arXiv:1201.1931 [hep-th];
V.~Bonzom and H.~Erbin,
  {\it Coupling of hard dimers to dynamical lattices via random tensors,}
  J.\ Stat.\ Mech.\  {\bf 1209} (2012) P09009
  [arXiv:1204.3798 [cond-mat.stat-mech]].

\bibitem{loopgas}
V.~Bonzom and F.~Combes,
  {\it Fully packed loops on random surfaces and the 1/N expansion of tensor models,}
  arXiv:1304.4152 [hep-th].

\bibitem{DarioRazvan}
D.~Benedetti and R.~Gurau,
  {\it Phase Transition in Dually Weighted Colored Tensor Models,}
  Nucl.\ Phys.\ B {\bf 855} (2012) 420
  [arXiv:1108.5389 [hep-th]].

\bibitem{branched}
R.~Gurau and J.~P.~Ryan,
  {\it Melons are branched polymers,}
  arXiv:1302.4386 [math-ph].


\bibitem{sdequations}
R.~Gurau,
  {\it A generalization of the Virasoro algebra to arbitrary dimensions,}
  Nucl.\ Phys.\ B {\bf 852} (2011) 592
  [arXiv:1105.6072 [hep-th]].

  \bibitem{sdequations1}
R.~Gurau,
  {\it The Schwinger Dyson equations and the algebra of constraints of random tensor models at all orders,}
  Nucl.\ Phys.\ B {\bf 865} (2012) 133
  [arXiv:1203.4965 [hep-th]].

  \bibitem{sdequations2}
  T.~Krajewski,
  {\it Schwinger-Dyson Equations in Group Field Theories of Quantum Gravity,}
  arXiv:1211.1244 [math-ph].

   \bibitem{sdequations3}
   V.~Bonzom,
 {\it Revisiting random tensor models at large N via the Schwinger-Dyson equations},
  arXiv:1208.6216 [hep-th].


\bibitem{universality}
 R.~Gurau,
  {\it Universality for Random Tensors,}
  arXiv:1111.0519 [math.PR].



\bibitem{colorless}
V.~Bonzom, R.~Gurau and V.~Rivasseau,
  {\it Random tensor models in the large N limit: Uncoloring the colored tensor models},
  Phys.\ Rev.\ D {\bf 85} (2012) 084037
  [arXiv:1202.3637 [hep-th]].




%
\bibitem{FG}
  M. ~Ferri and C.~Gagliardi
  {\it  Crystallisation moves},
  Pacific\ Journal\ of\ Mathematics\ Vol. 100, No. 1, 1982

\bibitem{RazvanLost}
  R.~Gurau,
 {\it Lost in Translation: Topological Singularities in Group Field Theory},
  Class.\ Quant.\ Grav.\  {\bf 27}, 235023 (2010)
  [arXiv:1006.0714 [hep-th]].

  \bibitem{jimmy}
  J.~P.~Ryan,
  {\it Tensor models and embedded Riemann surfaces},
  Phys.\ Rev.\ D {\bf 85}, 024010 (2012)
  [arXiv:1104.5471 [gr-qc]].








\bibitem{vincentjoseph1}
J.~Ben Geloun and V.~Rivasseau,
  {\it A Renormalizable 4-Dimensional Tensor Field Theory,}
  arXiv:1111.4997 [hep-th].

  \bibitem{joseph1}
J.~Ben Geloun,
  {\it Two and four-loop $\beta$-functions of rank 4 renormalizable tensor field theories,}
  Class.\ Quant.\ Grav.\  {\bf 29} (2012) 235011
  [arXiv:1205.5513 [hep-th]].

  \bibitem{joseph2}
J.~Ben Geloun,
  {\it Asymptotic Freedom of Rank 4 Tensor Group Field Theory,}
  arXiv:1210.5490 [hep-th].

\bibitem{joseph3}
  J.~Ben Geloun and D.~O.~Samary,
  {\it 3D Tensor Field Theory: Renormalization and One-loop $\beta$-functions},
  arXiv:1201.0176 [hep-th].


\bibitem{dansyl}
S.~Carrozza and D.~Oriti,
  {\it Bubbles and jackets: new scaling bounds in topological group field theories,}
  JHEP {\bf 1206} (2012) 092
  [arXiv:1203.5082 [hep-th]].

  \bibitem{dansyl1}
S.~Carrozza, D.~Oriti and V.~Rivasseau,
  {\it Renormalization of Tensorial Group Field Theories: Abelian U(1) Models in Four Dimensions,}
  arXiv:1207.6734 [hep-th].


  \bibitem{dansyl3}
   S.~Carrozza, D.~Oriti and V.~Rivasseau,
  {\it Renormalization of an SU(2) Tensorial Group Field Theory in Three Dimensions},
  arXiv:1303.6772 [hep-th].

\bibitem{eterajoseph}
  J.~B.~Geloun and E.~R.~Livine,
  {\it Some classes of renormalizable tensor models},
  arXiv:1207.0416 [hep-th].

  \bibitem{vignes}
  D.~O.~Samary and F.~Vignes-Tourneret,
  {\it Just Renormalizable TGFT's on $U(1)^{d}$ with Gauge Invariance},
  arXiv:1211.2618 [hep-th].

\bibitem{VincentMatteo}
J. Magnen, K. Noui, V. Rivasseau and M. Smerlak, {\it Scaling behaviour of three-dimensional group field theory},  Class.Quant.Grav. 26 (2009) 185012, arXiv:0906.5477 [hep-th]

\bibitem{RazvanLVE} R. Gurau,
{\it The 1/N Expansion of Tensor Models Beyond Perturbation Theory},  arXiv:1304.2666 [math-ph]




\bibitem{GFTcosmo}
S.~Gielen, D.~Oriti and L.~Sindoni,
  {\it Cosmology from Group Field Theory},
  arXiv:1303.3576 [gr-qc].



\bibitem{EteraWinston}
W.~J.~Fairbairn and E.~R.~Livine,
  {\it 3d Spinfoam Quantum Gravity: Matter as a Phase of the Group Field Theory},
  Class.\ Quant.\ Grav.\  {\bf 24} (2007) 5277
  [gr-qc/0702125 [GR-QC]].

\bibitem{DanieleEteraFlorian}  F.~Girelli, E.~R.~Livine and D.~Oriti,
  {\it 4d Deformed Special Relativity from Group Field Theories},
  Phys.\ Rev.\ D {\bf 81}, 024015 (2010)
  [arXiv:0903.3475 [gr-qc]].


\bibitem{JimmyEteraDaniele}
  E.~R.~Livine, D.~Oriti and J.~P.~Ryan,
 {\it Effective Hamiltonian Constraint from Group Field Theory},
  Class.\ Quant.\ Grav.\  {\bf 28} (2011) 245010
  [arXiv:1104.5509 [gr-qc]].



\bibitem{DanieleEmergentMatter}
D. Oriti, {\it Emergent non-commutative matter fields from Group Field Theory models of quantum spacetime}, J.Phys.Conf.Ser. 174 (2009) 012047, arXiv:0903.3970 [hep-th]


\bibitem{DanieleLorenzo}
D.~Oriti and L.~Sindoni,
  {\it Towards classical geometrodynamics from Group Field Theory hydrodynamics},
  New J.\ Phys.\  {\bf 13}, 025006 (2011)
  [arXiv:1010.5149 [gr-qc]].



\bibitem{joseph}
  J.~Ben Geloun,
 {\it Classical Group Field Theory},
  J.\ Math.\ Phys.\  {\bf 53}, 022901 (2012)
  [arXiv:1107.3122 [hep-th]].



\bibitem{AristideDanieleFlorian}
  A.~Baratin, F.~Girelli and D.~Oriti,
  {\it Diffeomorphisms in group field theories},
  Phys.\ Rev.\ D {\bf 83}, 104051 (2011)
  [arXiv:1101.0590 [hep-th]].



\bibitem{matrixdouble}
E.~Brezin and V.~A.~Kazakov,
  {\it Exactly Solvable Field Theories Of Closed Strings,}
  Phys.\ Lett.\ B {\bf 236} (1990) 144.

  M.~R.~Douglas and S.~H.~Shenker,
  {\it Strings in Less Than One-Dimension,}
  Nucl.\ Phys.\ B {\bf 335} (1990) 635.

  D.~J.~Gross and A.~A.~Migdal,
  {\it Nonperturbative Two-Dimensional Quantum Gravity,}
  Phys.\ Rev.\ Lett.\  {\bf 64} (1990) 127.


\bibitem{tensordouble}

R.~Gurau,
  {\it The Double Scaling Limit in Arbitrary Dimensions: A Toy Model,}
  Phys.\ Rev.\ D {\bf 84} (2011) 124051
  [arXiv:1110.2460 [hep-th]].














\end{thebibliography}
\end{document}